\documentclass[sigconf,noacm]{acmart}

\usepackage{amsmath}
\usepackage[noend]{algorithmic}
\usepackage{algorithm}
\usepackage{multicol}
\usepackage{amstext}
\usepackage{amsthm}
\usepackage{multirow}
\usepackage{balance}
\usepackage{hyperref}
\usepackage{tikz}
\usepackage[underline=false]{pgf-umlsd}
\usetikzlibrary{calc,positioning}

\renewcommand{\epsilon}{\varepsilon}

\hypersetup{final}

\newcommand{\eps}{\ensuremath{\varepsilon}}

\newcommand{\calA}{\ensuremath{\mathcal{A}}}
\newcommand{\calB}{\ensuremath{\mathcal{B}}}

\newcommand{\calD}{\ensuremath{\mathcal{D}}}

\newcommand{\calS}{\ensuremath{\mathcal{S}}}

\renewcommand{\Pr}{\mathop{\mathbf{Pr}}}

\newtheorem{lem}{Lemma}[section]

\newtheorem{thm}[lem]{Theorem}

\newtheorem{definition}[lem]{Definition}

\newcommand{\bracket}[1]{\left(#1\right)}

\makeatletter
\newcommand{\vast}{\bBigg@{4}}
\newcommand{\Vast}{\bBigg@{5}}
\makeatother

\newcommand{\ltwo}[1]{\left\|#1\right\|_2}

\tikzset{%
  every neuron/.style={
    circle,
    draw,
    minimum size=1cm
  },
  neuron missing/.style={
    draw=none, 
    scale=4,
    text height=0.333cm,
    execute at begin node=\color{black}$\vdots$
  },
}
\usepackage{booktabs}       %

\AtBeginDocument{%
  \providecommand\BibTeX{{%
    \normalfont B\kern-0.5em{\scshape i\kern-0.25em b}\kern-0.8em\TeX}}}

\begin{document}
\fancyhead{}
\title{Network Shuffling: Privacy Amplification via Random Walks}

\author{Seng Pei Liew}
\email{sengpei.liew@linecorp.com}
\orcid{0000-0003-2419-2505}
\affiliation{%
  \institution{LINE Corporation}
  \city{Tokyo}
  \country{Japan}
}
\author{Tsubasa Takahashi}
\email{tsubasa.takahashi@linecorp.com}
\orcid{0000-0002-0646-0222}
\affiliation{%
  \institution{LINE Corporation}
  \city{Tokyo}
  \country{Japan}
}
\author{Shun Takagi}
\email{takagi.shun.45a@st.kyoto-u.ac.jp}
\orcid{0000-0001-7732-2807}
\affiliation{%
  \institution{Kyoto University}
  \city{Kyoto}
  \country{Japan}
}
\author{Fumiyuki Kato}
\orcid{0000-0001-9276-4232}
\email{fumiyuki@db.soc.i.kyoto-u.ac.jp}
\affiliation{%
  \institution{Kyoto University}
  \city{Kyoto}
  \country{Japan}
}
\author{Yang Cao}
\email{yang@i.kyoto-u.ac.jp}
\orcid{0000-0002-6424-8633}
\affiliation{%
  \institution{Kyoto University}
  \city{Kyoto}
  \country{Japan}
}
\author{Masatoshi Yoshikawa}
\email{yoshikawa@i.kyoto-u.ac.jp}
\orcid{0000-0002-1176-700X}
\affiliation{%
  \institution{Kyoto University}
  \city{Kyoto}
  \country{Japan}
}

\begin{abstract}
Recently, it is shown that shuffling can amplify the central differential privacy guarantees of data randomized with local differential privacy.
Within this setup, a centralized, trusted shuffler is responsible for shuffling by keeping the identities of data anonymous, which subsequently leads to stronger privacy guarantees for systems.
However, introducing a centralized entity to the originally local privacy model loses some appeals of not having any centralized entity as in local differential privacy.
Moreover, implementing a shuffler in a reliable way is not trivial due to known security issues and/or requirements of advanced hardware or secure computation technology. 

Motivated by these practical considerations, we rethink the shuffle model to relax the assumption of requiring a centralized, trusted shuffler.
We introduce network shuffling, a decentralized mechanism where users exchange data in a random-walk fashion on a network/graph, as an alternative of achieving privacy amplification via anonymity.
We analyze the threat model under such a setting, and propose distributed protocols of network shuffling that is straightforward to implement in practice.
Furthermore,  we show that the privacy amplification rate is similar to other privacy amplification techniques such as uniform shuffling.
To our best knowledge, among the recently studied  intermediate trust models that leverage privacy amplification techniques, our work is the first that is not relying on any centralized entity to achieve privacy amplification. 
\end{abstract}



\maketitle

\begin{figure}
\begin{center}
    \includegraphics[width=0.45\textwidth]{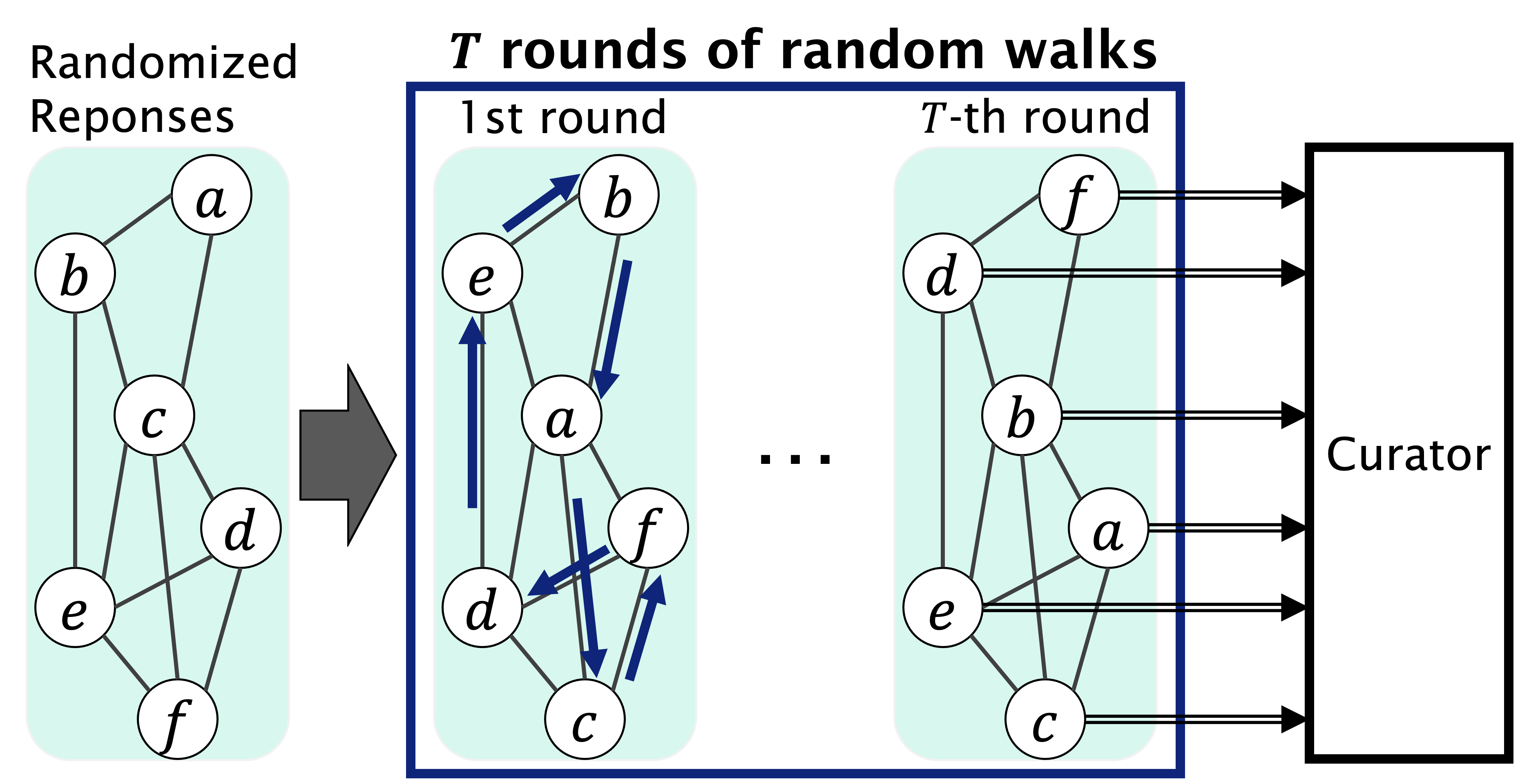}
\caption{Network shuffling is a decentralized mechanism of privacy amplification via multi-round random walks on a graph. In each round, every client relays her randomized reports to one of her neighbors (e.g., friends on a social network) via an encrypted channel. Network shuffling does not require any additional trusted entities unlike existing shuffling mechanisms.}
\label{fig:teaser}
\end{center}
\end{figure}

\section{Introduction}
Modern technology companies often gather data from large population of users/clients to improve their services.
Going beyond collecting data, more recently, users are asked to perform certain computation and send the reports back to the server, i.e., participating in federated learning or federated optimization \cite{DBLP:conf/aistats/McMahanMRHA17,DBLP:journals/ftml/KairouzMABBBBCC21,DBLP:journals/corr/abs-2107-06917}.
The reports/data involved can be sensitive in such distributed systems, and in order to protect the privacy of users, formal privacy guarantees are sought after.
For this purpose, differential privacy (DP) \cite{DBLP:conf/tcc/DworkMNS06,dwork2006our}, which is widely regarded as a gold-standard notion of privacy, has seen adoption by companies such as Google \cite{DBLP:conf/ccs/ErlingssonPK14}, Apple \cite{apple2017learning}, Linkedin \cite{DBLP:journals/corr/abs-2002-05839} and Microsoft \cite{DBLP:conf/nips/DingKY17}. 

Early academic studies of DP assume the existence of a centralized, trusted curator.
The centralized curator has access to raw data from all users, and is responsible for releasing the aggregated results with DP guarantees.
Such a central model requires users to trust the curator for handling their data securely.
However, many applications such as federated learning, or regulations like General Data Protection Regulation (GDPR) require different assumptions about the availability of such an entity. 
High-profile data breaches reported lately \cite{breach2020} have also given service providers second thought at collecting data in a centralized manner to avoid such risks. 

Alternative trust models without a centralized, trusted curator have been proposed.
Among these models, the model of local differential privacy (LDP) provides the strongest guarantees of privacy  \cite{DBLP:conf/pods/EvfimievskiGS03,DBLP:journals/siamcomp/KasiviswanathanLNRS11}: user does not have to trust any other entity except herself.
This is achieved by randomizing her data by herself using techniques such as randomized response to achieve LDP before aggregating the data.
However, LDP is known to be suffering from significant utility loss.
For example, the error of real summation with $n$ users under LDP is $\sqrt{n}$ larger than that of the central model \cite{DBLP:conf/esa/ChanSS12}.

Intermediate trust models between the local and central models are also considered in the literature.
These models aim to obtain better utility under more practical privacy assumptions \cite{DBLP:conf/stoc/DworkNPR10,Cheu_2019,soda-shuffling, DBLP:conf/sigmod/ChowdhuryW0MJ20}.
The pan-private model \cite{DBLP:conf/stoc/DworkNPR10} assumes that the curator is trustworthy for now, and store the distorted values to protect against future adversaries.
More recently, the shuffle model \cite{DBLP:conf/sosp/BittauEMMRLRKTS17} has attracted attention.
Within this model, it is assumed that there exists a centralized, trusted ``shuffler''. Users first randomizes their data before sending them to the shuffler, which executes permutation on the set of data to keep the identities of data anonymous.
Finally, the set of data is sent to the curator, on which no trust assumption is put.
Privacy amplification due to anonymity is shown to be achievable in this model, where it is enough for each user to apply a relatively small amount of randomizations/noises to the data to achieve strong overall DP guarantees \cite{Cheu_2019,soda-shuffling}.
Various aspects of shuffling have been studied in the literature \cite{DBLP:conf/crypto/BalleBGN19,DBLP:conf/icits/BalcerC20,DBLP:conf/aistats/GirgisDDKS21,DBLP:conf/aaai/Liu0CGY21,DBLP:conf/icml/Ghazi0MPS21,clones}.

Shuffle models provide better privacy-utility trade-offs, but they ultimately rely on a centralized entity; reintroducing such an entity seems to surrender the original benefits of not depending on any centralized entity in the local model, some of which have been discussed above (GDPR constraints, data breach risks). 
Furthermore, achieving anonymity via a centralized entity is not trivial in practice.
One realization of shuffler, Prochlo \cite{DBLP:conf/sosp/BittauEMMRLRKTS17}, requires, at its core, the use of a Trusted Executed Environment (TEE) such as Intel Software Guard Extension (SGX).
However, SGX is known to be vulnerable to side-channel attacks \cite{DBLP:journals/corr/abs-2006-13598}.
Another known realization is via a centralized set of mix relays, or mix-nets \cite{DBLP:journals/cacm/Chaum81,Cheu_2019}, but mix-nets remain vulnerable to individual relay failures and other issues \cite{DBLP:journals/comcom/RenW10}.

\subsection{Our Contributions}
In order to overcome the aforementioned weaknesses, we propose \textit{network shuffling}, a mechanism that achieves the effect of privacy amplification without requiring any centralized and trusted entity, in contrast to previous studies of privacy amplification where the existence of such an entity is assumed by default.

Essentially, within our framework, users exchange their reports in a random, secret, and peer-to-peer manner on a graph for multiple rounds. 
The random and secret exchange is vital at achieving anonymity of the reports: all users are potential original holders of each report after multiple exchanges.
The exchanged reports are subsequently sent to the curator. 
When this collection of reports is viewed in the central model, the privacy guarantees are enhanced, yielding privacy amplification via anonymity.
A simple illustration of network shuffling is shown in Figure \ref{fig:teaser}.

Network shuffling is inherently a decentralized approach.
By comparing it with existing centralized approaches, we show that network shuffling has several favorable properties in terms of security and practicality.
Concrete distributed protocols are also proposed to realize network shuffling.
Furthermore, to perform a formal privacy analysis of network shuffling, we model it as a random walk on graphs.
Harnessing graph theory, we give analytical results of privacy amplification, which have interesting relations with the underlying graph structure. 
Table \ref{tab:compare} compares the privacy amplification result with other techniques.
It can be seen that a privacy amplification of $O(1/\sqrt{n})$ similar to other techniques can be attained with network shuffling.
It is also the first general technique that achieves privacy amplification without any centralized entity as far as we know.

\begin{table}[t]
    \centering
    \begin{tabular}{ll}
    \toprule
    Mechanism & Privacy Amplification \\
    \midrule
    No amplification \cite{duchi2013local} & $\eps_0$\\
    Uniform subsampling \cite{DBLP:journals/siamcomp/KasiviswanathanLNRS11,DP_SGD} & $O(e^{\eps_0}/\sqrt{n})$\\
    Uniform shuffling \cite{soda-shuffling}& $O(e^{3\eps_0}/\sqrt{n})$\\
    Uniform shuffling (w/ clones) \cite{clones} & $O(e^{0.5\eps_0}/\sqrt{n})$\\ 
    \textbf{Network shuffling (ours)} & $O(e^{1.5\eps_0}/\sqrt{n})$\\
    \bottomrule
    \end{tabular}
    \caption{Comparisons of different privacy amplification mechanisms\protect\footnotemark. For network shuffling, see Theorem \ref{thm:all_protocol} for details.}
    \label{tab:compare}
\end{table}
\footnotetext{Here, the dependence on other factors such as $\delta$ is suppressed.}

We note that our main purpose is to show, in an application and system-agnostic manner, that privacy amplification in a decentralized and peer-to-peer manner is viable through network shuffling.
To make our arguments applicable in general, we necessarily make idealistic assumptions on the realization of the network shuffling mechanism. 
Some of the assumptions include non-colluding users and fault-tolerant communication.  
Such limitations and possible solutions will also be discussed in this work in later sections. 
Our contributions are summarized as follows:
\begin{itemize}
    \item We propose and motivate network shuffling, a practical, simple yet effective mechanism of privacy amplification that relaxes the assumption of requiring a centralized, trusted entity.
    We compare network shuffling with existing approaches to demonstrate its advantages (Section \ref{sec:motif}). 
    \item We formalize network shuffling as a random walk on graphs and propose minimal designs and distributed protocols that can be implemented in a simple and practical fashion (Section \ref{sec:proto}).
    \item Detailed privacy analysis and empirical evaluation are presented to show that the rate of privacy amplification is similar to that of uniform shuffling and uniform subsampling (Section \ref{sec:analysis} and \ref{sec:proofs}).
\end{itemize}

\section{Preliminaries and notations}

This section gives essential terminology, definitions, and theorems related to differential privacy, as well as relevant privacy amplification techniques for understanding our proposals.
Table \ref{tab:notations} gives some notations frequently used in this work.

\noindent
\textbf{Terminology.}
The following terminology is also used.
A \textit{user} or a \textit{node} refers to an entity in the system which holds a piece of information, referred to as \textit{report} or \textit{data}.
A user may send or receive reports from other users, of which the process is referred to as \textit{exchange} or \textit{relay}.
The exchange typically takes place for multiple \textit{rounds} or \textit{time steps}.
Users exchange reports on a \textit{path} or \textit{channel}, and two users are said to be \textit{connected} if they are able to exchange reports with each other on the path.
The system of users and paths forms a \textit{graph} or \textit{communication network}, or \textit{network} in short.
A \textit{curator} or \textit{server} will eventually collect all reports from the users.
At times we abstract away from how network shuffling is processed in practice and assume that the report anonymization procedure satisfies several minimal requirements, particularly when discussing its theoretical properties.
When we are discussing more practical details, such as protocols to achieve the requirements of report anonymization, we call them the \textit{implementation} or \textit{realization} of network shuffling.
Last but not least, we use ``log" to represent the natural logarithm.
\subsection{Differential Privacy}

\begin{definition}[($\epsilon$, $\delta$)-Differential Privacy]
Given privacy parameters $\epsilon \geq 0$ and $\delta \geq 0$, a randomized mechanism, $\mathcal{M}: \mathcal{D} \rightarrow \mathcal{S}$ with domain $\mathcal{D}$ and range $\mathcal{S}$ %
satisfies ($\epsilon$, $\delta$)-differential privacy (DP) if for any two adjacent inputs $D, D' \in \mathcal{D}$ and for any subset of outputs $S \subseteq \mathcal{S}$, the following holds:
\begin{align}
\Pr[\mathcal{M}(D) \in S] \leq e^\epsilon \cdot \Pr[\mathcal{M}(D') \in S] + \delta.
\end{align}
\end{definition}
Here, the notion ``adjacent'' is application-dependent.
For $D$ with $n$ elements, it refers to $D$ and $D'$ differing in one element.
For such cases, w.l.o.g., we also say that the mechanism is DP at index 1 in the central model.
$\delta$ is assumed to be smaller than $1/n$ and $\eps$ is assumed to be $\lesssim 1$ to provide meaningful central DP guarantees.
We also say that a $(\eps,\delta)$-DP mechanism satisfies \textit{approximate} DP, or $(\eps,\delta)$ indistinguishable when $\delta > 0$.

When $D$ consists of single element, $\mathcal{M}$ is also called a \textit{local randomizer}, which provides local DP guarantees. The formal definition is given below.
\begin{definition}[Local randomizer]
A mechanism $\mathcal{A}: \mathcal{D} \rightarrow \mathcal{R}$ is a $(\eps,\delta)$-DP local randomizer if for all pairs $x,x'\in \mathcal{D}$, $\mathcal{A}(x)$ and $\mathcal{A}(x')$ are $(\eps,\delta)$ indistinguishable.
\end{definition}

As $n$ elements satisfying $\eps$-LDP are naturally $\eps$-DP in the central model,
\textit{privacy amplification} is said to occur when the central model is $\eps'$-DP with $\eps' < \eps$.

\begin{table}
\begin{center}
\small
\begin{tabular}{ll}
        \toprule
        Symbol & Description \\
        \midrule
        $n$ & Total number of users\\
         $x_i$ & Report generated by user $i$\\
       $\calA_{ldp}$   & Local randomizer\\
        $\eps_0$  & LDP guarantee of local randomizer \\
                 $s_i$ & $i$-th randomized report\\
         $\calS^{(i)}$ & Domain of $i$-th randomized report \\
         $P^G_i$ & Probability distribution of user $i$ holding a report on a graph $G$ \\
         & (Position probability distribution)\\
         $\Gamma^G$&$\sum \limits_{i\in[n]}{n P^G_i}^2$, irregularity measure of graph $G$\\
        \bottomrule
    \end{tabular}
    \caption{Notations.}
    \label{tab:notations}
\end{center}
\end{table}

\subsection{Privacy Amplification Techniques}

We here introduce the shuffle model and other techniques of privacy amplification.

The shuffle model is a distributed model of computation, where there are $n$ users each holding report $x_i$ for $i \in [n]$, and a server receiving reports from users for analysis.
The report from each user is first sent to a shuffler, where a random permutation is applied to all the reports for anonymization. This procedure is also known as \textit{uniform shuffling} in the literature. 

Other mechanisms of privacy amplification have also been considered, such as privacy amplification by subsampling \cite{DBLP:journals/siamcomp/KasiviswanathanLNRS11}, which is utilized in federated learning \cite{DBLP:conf/iclr/McMahanRT018}.
However, it is necessary to put trust on the server to hide the identities of subsampled users to establish privacy amplification.

A technique called random check-in is introduced in \cite{DBLP:conf/nips/BalleKMTT20}, where a more practical distributed protocol within the federated learning setting is studied.
There, a centralized and trusted orchestrating server is still required to perform the ``check-ins'' of the users, which loses the appeal of our decentralized approach. 

Privacy amplification by decentralization has also been proposed \cite{DBLP:journals/corr/abs-2012-05326}. 
Deterministic (non-random) walking on graphs is considered there, and it is assumed that the central adversary does not have a local view of data (that is, the central adversary can only access the aggregated quantities).
In our work, we instead utilize random walks on graphs to achieve privacy amplification while maintaining a local view of data from the central adversary's perspective.
Due to privacy model's differences, \cite{DBLP:journals/corr/abs-2012-05326} is largely orthogonal to our work.

\section{Shuffling without a centralized, trusted shuffler}
\label{sec:motif}
We motivate the need for network shuffling by first discussing the properties of the current implementations of centralized shuffler.
Then, we discuss the general properties of network shuffling before closing the section by providing the threat model and assumptions of our proposal.

\subsection{Motivations for Network Shuffling}
Our work aims to resolve the shortcomings faced by existing realizations of the shuffle model, particularly Prochlo and mix-nets.
Let us first discuss in more detail the inadequate properties of the TEE-based Prochlo framework \cite{DBLP:conf/sosp/BittauEMMRLRKTS17}.
First, as mentioned above, it is difficult to avoid all types of side-channel attacks against TEEs \cite{DBLP:journals/corr/SchwarzWGMM17,DBLP:conf/uss/BulckMWGKPSWYS18,DBLP:journals/corr/abs-2006-13598}.
Second, Prochlo must first collect and batch reports from all users before performing shuffling concurrently.
This means that Prochlo is not scalable user number-wise.

Another common way of achieving anonymity is through the utilization of mix-nets \cite{DBLP:journals/cacm/Chaum81,DBLP:journals/popets/KwonLDF16}. 
Essentially, reports are relayed through a centralized set of mix relays (without batching) to avoid single point of failure. 
However, an adversary still can monitor the traffic to and from the mix-nets to determine the source of (un-batched) reports, thus breaking anonymity.
\textit{Cover traffic} (blending authentic messages with noises to counter traffic analysis) can alleviate the risk, but as will be discussed later in this section, network shuffling is more efficient than mix-nets traffic complexity-wise in this respect.
Moreover, the individual mix relays in turn become the obvious targets of attack and compromise \cite{DBLP:journals/comcom/RenW10,DBLP:conf/ccs/FreedmanM02,DBLP:conf/wpes/RennhardP02}.

The existence of a malicious insider located at the centralized shuffler (insider threat) is yet another potential risk that compromises the anonymity of reports.
Generally, introducing a centralized entity simply increases the attack surface.
These constraints, along with several other practical considerations discussed in Introduction, motivate us to consider alternatives which relax the assumption of the availability a centralized shuffler, but are still able to reap the benefits of shuffling, i.e., privacy amplification.

\noindent \textbf{Network Shuffling.}
Uniform shuffling is, simply put, a mechanism that mixes up reports from distinct users to break the link between the user and her report.
This is assumed in the literature to be executed by a centralized ``shuffler''.
The core idea of our proposal that removes the requirement of any centralized and trusted entity is simple: exchanging data among the users randomly to achieve the shuffling/privacy amplification effect, assuming that the users can communicate with the server and each other on a communication network.
We denote this mechanism by \textit{network shuffling}.

Initially, each user randomizes her own report using the local randomizer.
Then, the user sends out the report randomly to other connected  users while receiving incoming reports also from the connected users.
The above procedure is iterated for a pre-determined rounds before sending the reports to the curator. This procedure of exchanging reports essentially achieves the effect of hiding the origin of the report like what a centralized shuffler does.

Network shuffling is applicable to any group of users able to form a communication network to exchange reports with each other.
For example, it can be applied to users of messaging applications, where the report-exchanging network is naturally defined by the social network. 
It can also be applied straightforwardly to wireless sensor networks \cite{DBLP:conf/ical/ZhangZ12}, as well as the Internet Protocol (IP) network overlay which is general-purpose and application-transparent \cite{DBLP:conf/iptps/FreedmanSCM02,DBLP:conf/ccs/FreedmanM02,DBLP:conf/wpes/RennhardP02}.

\subsection{Complexity Analyses}
There are a few basic properties of network shuffling that can be inferred without specifying its detailed realizations.
First, in terms of memory, as each user exchanges reports without requiring to store them, the amount of memory or space taken is constant with respect to $n$ ($O(1)$).
This is in contrast to Prochlo ($O(n)$).
Mix-nets also have memory complexity of $O(1)$ at best, if no batching of data is made to the relay traffic.

Prochlo requires each user to send the report once, leading to a traffic complexity of $O(1)$ per user.
Typically, mix-nets utilize cover traffic to defend against traffic analysis such that the adversary cannot distinguish whether a report is genuine or merely noise.
To explicitly cover all of $n$ users, mix-nets must send cover traffic to all users, leading to a traffic complexity of $O(n)$ per user. 
In contrast, each user exchanges traffic only with her neighbors in network shuffling protocols when using cover traffic, and the number of neighbors is typically much smaller than $n$.
Some protocols require relaying $O(\log n)$ reports (ignoring other contributing factors; see Section \ref{subsec:setup}), amounting to a traffic complexity of $O(\log n)$, or $O(1)$ per user if additional relays depending on $n$ are not required.

Finally, we note that as TEE is limited in terms of memory, shuffling is processed in batches of reports, requiring multiple rounds of processing \cite{DBLP:conf/sosp/BittauEMMRLRKTS17}.
On the other hand, mix-nets and network shuffling at minimum require only simple encryption-decryption mechanisms (see Section \ref{subsec:comm}).
Overall, although requiring at most $O(\log n)$ of traffic overhead, network shuffling has minimal memory and processing overhead consumption.
The dependence of both memory and traffic complexities are summarized in Table \ref{tab:comp}.

We would like to emphasize that the focus of this paper is on the general, application or system-agnostic study of privacy amplification with network shuffling.
A full-fledged implementation and system analysis would be specific to the underlying application or system, and is therefore out of the scope of this work. \footnote{See \cite{DBLP:conf/iptps/FreedmanSCM02,DBLP:conf/ccs/FreedmanM02,DBLP:conf/wpes/RennhardP02} for an implementation and analysis of a similarly decentralized and peer-to-peer anonymization system that is specific to IP network overlay.} 

Nevertheless, for the rest of this section and next, we describe a minimal and self-contained implementation of network shuffling, along with some idealistic assumptions to abstract away from application or system-specific details, but is still realistic enough to be implemented.
We hope that this will help practitioners understand the minimal requirements to achieve privacy amplification via network shuffling, which is, as far as we know, the first decentralized approach of privacy amplification in the literature. 
Furthermore, we discuss in detail in later sections where there are rooms for refinements in our descriptions for real-world deployment. 

We begin by elaborating on our threat model which forms the basis of the network shuffling protocols provided in Section \ref{sec:proto}.

\begin{table}
\begin{center}
\small
    \begin{tabular}{cccc}
        \toprule
        
        Complexity & Prochlo& Mix-nets& \textbf{Ours} \\
                \midrule
       Entity space complexity& $O(n)$& $O(1)$& $O(1)$ \\
       User traffic complexity& $O(1)$& $O(n)$& $O(\log n)/O(1)$ \\
        \bottomrule
    \end{tabular}
    \caption{Complexity comparisons between our proposal, network shuffling and other centralized approaches.
    Entity space complexity refers to the amount of memory space required to process shuffling by, the shuffler (Prochlo), and the user (mix-nets and network shuffling).
        User traffic complexity refers to the number of report exchanged/sent by each user.
    Dependence on factors other than the number of users, $n$, is suppressed.}
    \label{tab:comp}
\end{center}
\end{table}

\subsection{Our Threat Model and Assumptions}
\label{subsec:threat}
As all users apply an $\eps_0$-DP local randomizer to the report before sending it out, they are guaranteed with local DP quantified by $\eps_0$. 
This is true even when all other parties (including other users) collaborate to attack any specific user.
This forms the worst-case privacy guarantees of network shuffling.

Our privacy analysis in the following sections is however mainly against a central analyzer. 
In order to establish central DP, our threat model requires additional assumptions.

\noindent\textbf{Non-collusion}. We assume that there is no collusion among users, that is, users do not collude against certain victim user.
Note that this assumption is also required by the shuffle model when privacy amplification via shuffling is considered \cite{DBLP:journals/pvldb/WangXDZHHLJ20}.

\noindent\textbf{Honest-but-curious users}.
All users are assumed to be honest but can be curious.
That is, users will not deviate from the protocol but can try to retrieve information from the received reports.
We will demonstrate a communication protocol in Section \ref{subsec:comm} to verify that this is achievable in practice.
Besides, we assume that all users are available to participate in relaying the reports at each round.
Relaxation of these assumptions will be discussed in Section \ref{subsec:prac}.

\noindent\textbf{No traffic analysis}.
For simplicity, we further assume that it is not possible for the adversary to perform timing or traffic analysis.
Achieving this in practice may require users to send more reports than necessary to cover the traffic as described earlier, or hide the trails by sending the report along with other information. 
Our threat model abstracts away from these considerations.
Nevertheless, we remark the final-round reports are not anonymous: the central adversary is able to link the report to the user sending it at the final round of network shuffling (but not the ``original" owners of the reports).
This is realistic when considering communication protocols in practice, such as those given in Section \ref{sec:proto}. 

When the above assumptions on the central adversary fail to hold, the privacy guarantees degrade at worst to $\eps_0$, privacy guarantees in the LDP setting. 
We also note that our threat model is reasonable compared to shuffle model's.
That is, for both cases, when all users except the victim collude with the server, or when traffic analysis is possible, the privacy guarantees drop to the LDP setting.
Within the uniform shuffling framework, the server may additionally collude with the shuffler.
On the other hand, network shuffling offers an alternate privacy amplification solution without this attack surface.

\section{Protocols of Network Shuffling}
\label{sec:proto}
In order to perform privacy analysis of network shuffling, we model the report exchanges as random walks on graphs. 
In the following, we first provide notions of random walks, focusing on relating them to applications in network shuffling.
We will also quote well-known results from graph theory, particularly those related to the study of random walks on graphs.
See \cite{Lovasz1996} for a relevant comprehensive survey of random walks on graphs.
Then, we describe the distributed protocols of network shuffling.
\subsection{Random Walks on Graphs}
\label{sec:rdmgrph}

A graph, $G = (V,E)$, is characterized by a set of nodes or vertices, $V$, and a set of paired nodes, $E$.
In our case, $G$ is the communication network, and users may be viewed as nodes and $E$ represents the set of communication paths between users.
In this work, we consider undirected graph, that is, each pair of connected users can send messages to each other.

A particular graph topology is of interest: $k$\textit{-regular graph}. 
A $k$-regular graph is a graph in which each node has the same number, $k$, of connected neighbors.

We define the \textit{adjacency matrix} of $G$ as $A\in \mathbb{R}^{n\times n}$, where $n$ is the number of nodes in $G$. $A_{ij}$ takes values from $\{0,1\}$, indicating whether the nodes $i$ and $j$ are connected. Then, the probability of user $i$ sending report to a randomly chosen recipient $j$ is characterized by the transition probability from node $i$ to node $j$ on the graph, 
$$M_{ij} =\frac{A_{ij}}{\sum_{j\in V}A_{ij}}.$$

Denote $B$, which represents a diagonal matrix, by $B_{ii}=\sum_{j\in V}A_{ij}$.
Then, we can rewrite the \textit{transition probability matrix} $M$ of graph $G$ as $M=AB^{-1}$.
We write $P(t) \in \mathbb{R}^n$ as the probability distribution of the $n$ users holding a report at time $t$ (we also refer to it as a \textit{position probability distribution}).
Then, the update of probability distribution due to random exchange (random walk) of report at each time step may be expressed recursively as $P(t+1)=M^{T}P(t).$
Given a certain initial probability distribution, $P(0)$, we can calculate the probability distribution after $t$ time steps, $P(t)$ as $P(t)=(M^{T})^{t}P(0)$.
At times, we abbreviate $P^G(t)$ of graph $G$ as $P^G$, and its $i$-th component as $P^G_i$ for notational convenience. 

We are also interested in the probability distribution in the long run. \textit{Stationary distribution}, $\pi^G$ characterizes such a behavior:
\begin{definition}[Stationary distribution]
A probability distribution $\pi^G$ over a set of nodes $V$ on a graph $G=(V,E)$ is a stationary distribution of the random walk when $\pi^G = M^T\pi^G$.
\end{definition}
When a random walk converges to a stationary distribution eventually, we say that the walk is \textit{ergodic}:
\begin{definition}
A random walk is ergodic when for all initial probability distribution $P_0$ over a set of nodes $V$ on a graph $G=(V,E)$, $P_0$ converges to a stationary distribution as $t\to \infty$.
\end{definition}
The following theorem describes conditions for guaranteeing ergodicity of a random walk on a graph :
\begin{thm}
A random walk on a graph $G$ is ergodic if and only if $G$ is connected and not bipartite.
\end{thm}
Let $\mathbf{k}=(k(1), \dots, k(n))$ be the number of edges connected to each node, and $m$ be the total number of edges for an ergodic graph $G$.
It can further be shown that $\mathbf{k}/2m$ is a stationary distribution.
A regular graph's stationary distribution is therefore $\mathbf{1}/n$.

One is also interested in the \textit{mixing time}, the number of time steps it takes for a probability distribution to be close to the stationary distribution.
To estimate the mixing time, we first consider a $k$-regular graph and introduce several more results from spectral graph theory. \footnote{We mainly follow the arguments given in \cite{graphlecture}.
More details can be found there and in the references therein.}
First, the transition probability matrix $M=\frac{1}{k}A$ (equivalent to $B^{-1/2}AB^{-1/2}$, known as the normalized adjacency matrix) is characterized by its eigenvalues, denoted by $1 = \alpha_1 \geq \alpha_2 \geq \dots \alpha_n > -1$ and its corresponding orthonormal eigenvectors, $e_1, e_2, \dots e_n$.
It can be shown that $\alpha_2 < 1$ and $e_1 = \mathbf{1}/\sqrt{n}$.
Since $e_i$ for $i \in [n]$ is orthonomal in $\mathbb{R}^n$, we can write any initial distribution as $P_0 = \sum_{i=1}^n c_i e_i$ with $\sum_{i=1}^nc_i=1$.
Then, 
\begin{equation*}
P(t) = (M^T)^t P(0) =  \sum_{i=1}^n (M^T)^t c_i e_i  = \sum_{i=1}^n \alpha^t_i c_i e_i,
\end{equation*}
and $\lim_{t\to \infty} \sum_{i=1}^n \alpha^t_i c_i x_i = c_1\alpha_1 e_1$ using $\alpha_2 < 1$.
Since $\alpha_1=1$, $c_1=1/\sqrt{n}$ as $e_1 = \mathbf{1}/\sqrt{n}$.

We also quantify the concept of walk convergence as follows.
\begin{definition}
The graph total variation distance between two distributions $P$ and $Q$ on a graph $G$ is defined as:
\begin{equation}
\label{def:tv}
    TV_G(P,Q) \coloneqq \sum_{i=1}^n |P_i - Q_i| = \left\| P-Q \right\|_1.
\end{equation}
\end{definition}
Recall that the distribution at time $t$ is $P(t) = (M^T)^t P(0)$.
Let $\alpha \coloneqq \text{min} (1-\alpha_2, 1-|\alpha_n|)$ be the spectral gap.
It can be shown that:
\begin{align}
    TV_G(P(t)^G,\pi^G) &= \left\|(M^T)^t P(0) - \mathbf{1}/n\right\|_1 \nonumber \\
&= \left\|\sum_{i=2}^n c_i \alpha_i^t e_i\right\|_1 \leq \sqrt{n} \left\|\sum_{i=2}^n c_i \alpha_i^t e_i\right\|_2
\end{align}
using Cauchy-Schwartz inequality. Note that 
\begin{align}
\label{eq:l2}
    \sum_{i=2}^n c_i^2 \alpha_i^{2t} &\leq (1-\alpha)^{2t} \sum_{i=2}^n c_i^2  \\
    &\leq (1-\alpha)^{2t} \sum_{i=1}^n c_i^2 \leq (1-\alpha)^{2t}, \nonumber
\end{align}
where we have used $1 =\sum_{i=1}^n c_i \geq \sum_{i=1}^n c^2_i$ by Minkowski inequality.
Then, we have
\begin{align*}
    TV_G(P(t)^G,\pi^G) \leq \sqrt{n}(1-\alpha)^t.
\end{align*}

This means that when choosing $t\simeq \alpha^{-1}\log n$, we have
\begin{align}
    TV_G(P(t)^G,\pi^G) \leq \sqrt{n}e^{-\alpha t} 
    \lesssim 1/\sqrt{n}, 
\end{align}
where we have used $1-x \leq e^{-x}$.
Therefore, we can say that when the mixing time is of $O(\alpha^{-1}\log n)$, the graph total variation distance is small for sufficiently large $n$, and the probability distribution is close to the stationary distribution.
Finally, as $AB^{-1}$ is \textit{similar} to the corresponding normalized adjacency matrix for non-regular graphs, non-regular graphs also have the same convergence behavior \cite{graphlecture}.

\subsection{Network Shuffling as a Random Walk on Graphs}
\label{subsec:setup}
Let us elaborate on the setups of network shuffling in terms of graph theoretical notions introduced in Section \ref{sec:rdmgrph}.

We consider only connected graphs in our analysis.
The privacy of disconnected graphs may be viewed as a parallel composition of the privacy of connected sub-graphs, meaning that shuffling occurs only within each connected sub-graphs.
It is then sufficient to analyze connected graphs only.

We also assume that all users on the graph participate in network shuffling.
That is, all users are required to participate in receiving and sending reports to neighboring users.
Before the process of exchanging reports starts, each user is also required to have produced a randomized report to be exchanged: user $i$ for $i\in [n]$ produces $\calA_{ldp}(x_i)=s_i$.
After the final round of random exchanges, a user can possibly have received no report, a single report, or multiple reports. 
The set of reports held by user $i$ is denoted by $\{s_j\}^i_{j\in[n]}$.
Figure \ref{fig:neuron} illustrates how the reports are distributed to the users at different time steps.

We analyze separately the following two scenarios at any time step:
\begin{itemize}
    \item \textbf{Stationary distribution}: $P^G(t)$ of ergodic graph. 
    \item \textbf{Symmetric distribution}: $P^G(t)$ of $k$-regular graph. 
\end{itemize}
The technical reason for this separation is due to the their difference in the dynamics of $P^G(t)$ at time $t$, to be explained in detail in Section \ref{sec:finite}.

The first scenario concerns with the modeling of any connected and non-bipartite graphs which are best analyzed (as will be shown later) with respect to the stationary distribution they converge to, hence the name \textit{stationary distribution}.

The second scenario concerns with a walk that is symmetric over all nodes. 
Under this setup, one can w.l.o.g. analyze the privacy with respect to the first user, allowing for precise privacy analysis.
Note that such a consideration is not purely of theoretical interests.
Certain design \cite{DBLP:conf/iptps/FreedmanSCM02,DBLP:conf/ccs/FreedmanM02,DBLP:conf/wpes/RennhardP02} implements peer discovery protocols where each user proactively selects a set of users to form a path through the network. 
In this case, forming a $k$-regular graph is a reasonable consideration when each user uses the same peer discovery protocol to select a fixed $k$ number of other users to communicate with.
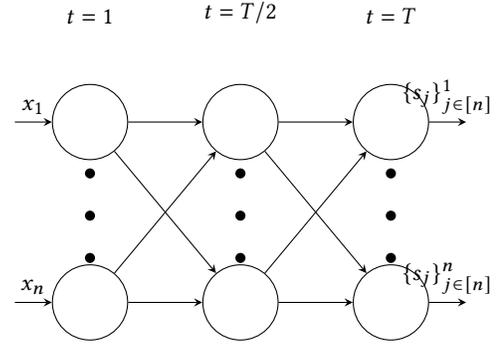
\begin{figure}
\begin{center}
{\begin{tikzpicture}[x=1cm, y=1.2cm, >=stealth]

\foreach \m/\l [count=\y] in {1,missing,2}
  \node [every neuron/.try, neuron \m/.try] (input-\m) at (0,2-\y) {};

\foreach \m [count=\y] in {1, missing,2}
  \node [every neuron/.try, neuron \m/.try ] (hidden-\m) at (2,2-\y) {};

\foreach \m [count=\y] in {1,missing,2}
  \node [every neuron/.try, neuron \m/.try ] (output-\m) at (4,2-\y) {};

\foreach \l [count=\i] in {1,n}
  \draw [<-] (input-\i) -- ++(-1,0)
    node [above, midway] {$x_\l$};

\foreach \l [count=\i] in {1,n}
  \draw [->] (output-\i) -- ++(1,0)
    node [above, midway] {$\{s_j\}^\l_{j\in[n]}$};

\foreach \i in {1,...,2}
  \foreach \j in {1,...,2}
    \draw [->] (input-\i) -- (hidden-\j);

\foreach \i in {1,...,2}
  \foreach \j in {1,...,2}
    \draw [->] (hidden-\i) -- (output-\j);

\foreach \l [count=\x from 0] in {1, T/2, T}
  \node [align=center, above] at (\x*2,2) {$t=\l$};
\end{tikzpicture} 
\caption{Graphical representation of the reporting protocol, with each node represents a user (only the first and the $n$-th users are shown above).
Each user starts with her own report ($t=1$).
Then, for each time step (only $t=T/2$ is shown above), the user sends the randomized report to other connected users.
At the final communication round ($t=T$), each user sends the randomized report(s) held by her to the server according to the protocols introduced in text.}
    \label{fig:neuron}}
\end{center}
\end{figure}

\subsection{Reporting Protocols}
\label{subsec:protocol}
We propose two user protocols of sending reports.
Note that we first abstract away from security considerations of the protocol which are discussed in Section \ref{subsec:comm}.
With this abstraction, we focus on the mechanism of sending the reports to the server by users.

\noindent\textbf{``All'' protocol}. The first protocol is described in Algorithm \ref{alg:a_all}.
Here, each user exchanges the reports in a random-walk manner for a pre-determined number of communication rounds.
Afterwards, the user sends all reports held by her to the server.
Note that the user simply sends a null response to the server if no report is on her hand during the final round.
This protocol is referred to as $\calA_{all}$.

\begin{algorithm}[t]
\begin{algorithmic}[1]
\caption{$\calA_{all}$ - Client-side protocol for user-$j$.}
\label{alg:a_all}
\STATE \textbf{Parameters:} local randomizer $\calA_{ldp}$, number of steps $t$, set of neighboring clients $A_j$ \\[0.5em]
\STATE $s_j \leftarrow \calA_{ldp}(x_j)$
\STATE $Z_j \leftarrow \{s_j\}$
\FOR{$i \in [t]$}
\FOR{$k \in S_j$}
\STATE Sample client $l \xleftarrow{u.a.r} A_j$ 
\STATE Send report $k$ to client $l$
\ENDFOR
\STATE Receive reports from neighboring clients and add them to $S_j$.
\ENDFOR
\IF{$S_j$ is empty}
\STATE \textbf{return} null
\ELSE
\STATE \textbf{return} $S_j$
\ENDIF
\end{algorithmic}
\end{algorithm}

\noindent\textbf{``Single'' protocol}. The second protocol is a modification of Algorithm \ref{alg:a_all} to provide better privacy guarantees, inspired by the federated learning approach in \cite{DBLP:conf/nips/BalleKMTT20}.
As before, users exchange the reports in a random-walk manner for a pre-determined number of communication rounds.
After the final communication round, if $\{s_j\}^i_{j\in[n]}$ is empty, a dummy report is sent to the server by user $i$.
Otherwise, a report is sampled uniformly from $\{s_j\}^i_{j\in[n]}$ to be sent to the server.
Intuitively, since each user sends only a single report irrespective of how many reports she has received, it is harder for the adversary to infer about the identity of the received reports, providing stronger privacy guarantees.
Note that as a trade-off, not sending all reports to the server could induce utility loss.
The protocol is described in Algorithm \ref{alg:a_single}, and is referred to as $\calA_{single}$.

\begin{algorithm}[t]
\begin{algorithmic}[1]
\caption{$\calA_{single}$ - Client-side protocol for user-$j$.}
\label{alg:a_single}
\STATE \textbf{Parameters:} local randomizer $\calA_{ldp}$, number of steps $t$, set of neighboring clients $A_j$ \\[0.5em]
\STATE $s_j \leftarrow \calA_{ldp}(x_j)$
\STATE $S_j \leftarrow \{s_j\}$
\FOR{$i \in [t]$}
\FOR{$k \in S_j$}
\STATE Sample client $l \xleftarrow{u.a.r} A_j$ 
\STATE Send report $k$ to client $l$
\ENDFOR
\STATE Receive reports from neighboring clients and insert them to $S_j$.
\ENDFOR
\IF{$S_j$ is empty}
\STATE $J_j \leftarrow \calA_{ldp}\left(\mathbf{0} \right) $ \hfill // Dummy report
\ELSE
\STATE Sample $J_j \xleftarrow{u.a.r.} S_j$
\ENDIF
\STATE \textbf{return} $J_j$
\end{algorithmic}
\end{algorithm}

\subsection{Communication Protocol}
\label{subsec:comm}
As a concrete use case to motivate this work, let us consider collecting data from users of instant messanging apps offered by social networking services, such as Whatsapp, or Messenger offered by Facebook.
Messages are commonly exchanged between two or more parties, where the sender sends messages to recipients connected to a centralized network run by a server.
Within the social networking setting, users commonly interact only with other users connected on the social network. 
It is then natural to consider exchanging the private reports only with users connected on the social network to hide the private reports within the traffic of usual data exchanges.
We note that user-to-user communication via a centralized server is not strictly required in other applications such as wireless sensor or Internet of Things (IoT) networks, where peer-to-peer (P2P) communication is feasible.

Let us consider the communication protocol running between the server and $n$ users.
The server is assumed to be able to communicate with the users via a secure, authenticated and private channel.
We also assume the existence of a Public Key Infrastructure (PKI), which ensures that only authenticated users can participate in the data exchange.
All users utilize two types of public-private keypairs. One is for end-to-end encrypted communication with other users ($<c_1^{pk},c_1^{sk}>$), which is unique to each user, and the other is for encrypting the report to be exchanged with other users, where the server holds the corresponding private key ($<c_2^{pk},c_2^{sk}>$).

All users initially publish and receive public keys via the PKI.
A user then applies local randomizer to the report, and encrypts it with $c_2^{pk}$.
Subsequently, the user uses $c_1^{pk}$ to send the report to another user in an end-to-end encrypted manner.
After exchanging the reports for a number of rounds of communication, the user sends the reports (encrypted only with $c_2^{pk}$) to the server.
The server then decrypts all the received reports and perform data analysis on them.
Notice that the server can link a user to her last received reports.
The illustration is shown in Figure \ref{fig:compro}.
We next analyze the security properties of the protocol.

\noindent\textbf{Security against adversarial server}. The use of $<c_1^{pk},c_1^{sk}>$ prevents the exposure of the encrypted report to anyone else other than the communicating users.
This end-to-end encryption especially protects the report's privacy from the possibly adversarial server.

\noindent\textbf{Security against honest-but-curious users}. The use of $<c_2^{pk},c_2^{sk}>$ prevents the exposure of the content of the report to anyone else other than the user applying the local randomizer and the server.
This protects the report's privacy from honest-but-curious users.

\begin{figure}
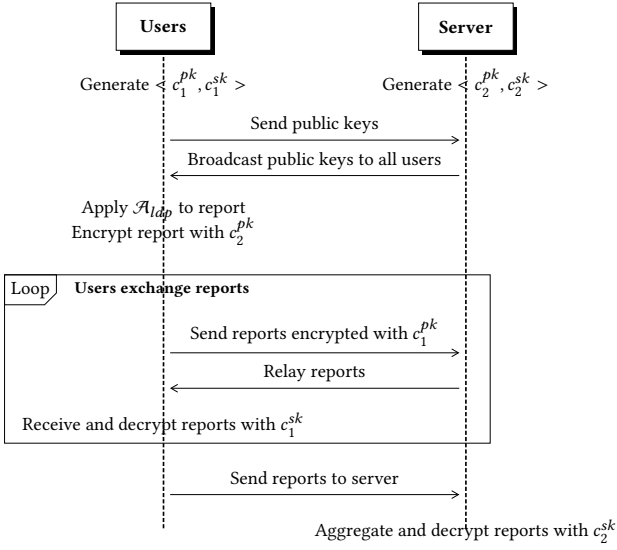

\begin{center}
\resizebox{0.48\textwidth}{!}
{
\begin{sequencediagram}
\tikzstyle{dotted}=[thick,dash pattern=on 2pt off 1pt]
\newinst{ue}{\textbf{Users}}
\newinst[3.5]{sv}{\textbf{Server}}
\node [below=3mm](t0) at (ue) {Generate $<c_1^{pk},c_1^{sk}>$};
\node [below=3mm](t00) at (sv) {Generate $<c_2^{pk},c_2^{sk}>$};
\postlevel
\mess{ue}{Send public keys}{sv}
\mess{sv}{Broadcast public keys to all users}{ue}
\node[below=3mm](t1) at (mess to){Apply $\calA_{ldp}$ to report};
\node[below=6mm](t2) at (mess to){Encrypt report with $c_2^{pk}$};
\postlevel
\postlevel
\begin{sdblock}{Loop}{\textbf{Users exchange reports}}
\mess{ue}{Send reports encrypted with $c_1^{pk}$}{sv}
\mess{sv}{Relay reports}{ue}
\node[below=3mm](t3) at (mess to){Receive and decrypt reports with $c_1^{sk}$};
\end{sdblock}
\postlevel
\mess{ue}{Send reports to server}{sv}
\node[below=3mm](t3) at (mess to){Aggregate and decrypt reports with $c_2^{sk}$};
\end{sequencediagram}}
\caption{Communication protocol for network shuffling based on instant messaging apps.}
\label{fig:compro}
\end{center}
\end{figure}

We would like to emphasize that our communication protocol based on asymmetric encryption is simple to implement in contrast to secure aggregation \cite{DBLP:conf/ccs/BonawitzIKMMPRS17} or secure shuffling, where sophisticated secure multi-party computation protocols or TEEs are required, which can be challenging in terms of implementation.

\subsection{Practical Considerations}
\label{subsec:prac}
The network shuffling protocols described thus far are self-contained and secure if the threat model and  assumptions given in Section \ref{subsec:threat} are satisfied.
Here, we consider scenarios where some of the assumptions are relaxed.
While the detailed study is beyond the current scope, we discuss potential solutions or workarounds by giving reference to relevant works in the literature.

\noindent{\textbf{Fault tolerance.}
In practice, users may fail to operate properly due to various reasons. They may disconnect from the network temporarily due to, e.g., battery depletion or network outage.
One way to model such a situation is via the use of \textit{lazy random walk}.
A lazy random walk is a random walk where at each time step, the walk has a certain probability of staying at its current node instead of transitioning to other nodes.
This behavior reflects the probability of certain users being disconnected temporarily and unable to exchange reports at a certain time.
Another approach is to analyze random walks on a dynamic graph to study the effects of potential walker losses \cite{DBLP:journals/tc/ZhongSS08}.

\noindent\textbf{Collusion.}
Colluding users can be a threat to the anonymity guarantees.
Discussing this threat however requires more assumptions on the system in use. 
For example, implementations at the network layer of IP \cite{DBLP:conf/iptps/FreedmanSCM02,DBLP:conf/ccs/FreedmanM02} defend against such an adversary by requiring the node to select peers in a pseudo-random way: it is highly unlikely for an adversary to control all pseudo-randomly selected nodes in a path.
Another defensive method is to monitor user behavior and use collusion detection algorithms to counter such an adversary \cite{DBLP:conf/wpes/RennhardP02}.
Moreover, the orchestrating server may drop users considered to be adversarial; such a dynamic scenario may be analyzed with a dynamic graph as discussed above.

\section{Privacy Analysis}
\label{sec:analysis}
In this section, we first describe useful preparations for analyzing the privacy of network shuffling.
Then, we study privacy theorems for the scenarios and protocols in consideration along with the proof sketches and interpretations.
Some of the technical analyses are relegated to Section \ref{sec:proofs}.
Numerical evaluations involving the use of real-world datasets are also given.

\subsection{Preparations}
The following notations describing the distance between distributions are convenient when illustrating the proof. Given two distributions, $\mu$ and $\mu'$ that are $(\eps, \delta)$-DP close, i.e., for all measurable outcomes $S$, the following holds:
$$e^{-\eps} \bracket{\mu'(S) - \delta} \leq \mu(S) \leq e^{\eps}\mu'(S) + \delta, $$
we denote this relation by $\mu \approxeq_{(\eps, \delta)}\mu'$. 

We will make use of the heterogeneous advanced composition for DP \cite{KOV17}. For a sequence of mechanisms, $\calA_1, \ldots, \calA_k$ which are $\eps_1, \ldots, \eps_k$-DP each, the $k$-fold adaptive composition is $\left(\eps, \delta\right)$-DP for any $\delta \in (0, 1)$, and 
\begin{equation}
    \eps = \sum\limits_{i\in[k]} \frac{(e^{\eps_i} - 1) \eps_i}{e^{\eps_i} + 1} + \sqrt{2 \log{\frac{1}{\delta}} \sum\limits_{i\in[k]}\eps_i^2}. \label{eqn:heta}
\end{equation}

We next provide several other results useful for proving the main theorems of this paper.

\begin{lem}\label{lem:allocate}
Let $L = (L_1, \ldots, L_n)$ denote the number of reports each of $n$ users is allocated in the protocol from Figure~\ref{alg:a_all}.
Also let $P^G_1, \ldots, P^G_n$ be the allocation probabilities.
With probability at least $1 - \delta$, we have $$\ltwo{L} \leq 
     \sqrt{(n^2-n)\sum_{i \in [n]}{P_i^G}^2}+ \sqrt{n \log(1/\delta)}.$$
\end{lem}
The proof is based on McDiarmid's inequality and is omitted here due to space constraint.

The following lemma is useful for extending the study of pure DP to approximate DP via the total variation distance $TV(P,Q) \coloneqq \sup\limits_{x\subset E} |P(x) - Q(x)|$, where $P,Q$ are probability distributions on $E$ (see also Definition \ref{def:tv}):
\begin{lem}[Balle et al., Cheu et al.]\label{lem:cheu}
Suppose $\calA : \calD \rightarrow \calS$ is an $(\eps_0, \delta_0)$-DP local randomizer with $\delta_0 \leq \frac{(1 - e^{-\eps_0}) \delta_1}{4 e^{\eps_0} \left(2 + \frac{\ln(2/\delta_1)}{\ln(1/(1-e^{-5\eps_0}))}\right)}$.
Then there exists an $8 \eps_0$-DP local randomizer $\tilde{\calA} : \calD \rightarrow \calS$ such that for any $x \in \calD$ we have $TV(\calA(x), \tilde{\calA}(x)) \leq \delta_1$.
\end{lem}
\begin{proof}
See Lemma A.3 of \cite{DBLP:conf/nips/BalleKMTT20}.
\end{proof}

In order to study shuffling, we make use of the technique introduced in \cite{soda-shuffling}, which reduces shuffling ($\calA_{sl}$) to swapping ($\calA_{swap}$) by swapping the first element with another element selected u.a.r. from the dataset, and applies the local randomizers.

Our privacy analysis concerns with how the adversary deduces the underlying data identities (assuming that $P^G$ is known) by making observations on the distribution of reports at the final time step, $\{s_j\}^i_{j\in[n]}$, as in Figure \ref{fig:neuron}.
At the heart of the proof is the reduction of the protocol to variants of the swapping algorithm \cite{soda-shuffling}.

Compared to uniform shuffling, analyzing the privacy of network shuffling faces a few additional challenges.
To understand this, let us first describe the technique \cite{soda-shuffling} of analyzing uniform shuffling. 

The uniform shuffling mechanism may be considered as a sequence of algorithms of the form $s_i \leftarrow \calA^{(i)}_{ldp}(s_{1:i-1}; x_{\pi(i)})$.
W.l.o.g, consider two datasets $D,D'$ differing in the first element.
For any permutation, one may first permute the last $n-1$ elements.
Then, the first element is swapped with an element uniformly sampled from the $n$ indices.
It is easy to see that this is equivalent to performing uniform permutation, and hence reduces uniform shuffling to swapping.
As each output has certain probability $p$ of being swapped with the first element, its distribution can be seen as a mixture of distribution, $\mu = (1-p)\mu_0 + p \mu_1$, where $\mu_0$ is independent of the first element (not swapped by the first element), and $\mu_1$ is the output distribution of $x_1$ randomized by the $i$-th local randomizer (swapped by the first element).
Then, one can show that $\mu_0$ and $\mu_1$ are overlapping mixtures to achieve the desired amplification (like subsampling \cite{BBG18}).
The amplified $\eps_i$ is then obtained by bounding $p$.
Finally, the heterogeneous composition theorem \cite{KOV17} is used to compose all $\eps_i$'s to obtain the overall $\eps$. 

Network shuffling is different in at least three ways. 
First, the adversary can link the reports $\{s_j\}^i_{j\in[n]}$ to the user last receiving them, different from uniform shuffling, where the link of the reports to the user is completely broken (randomized). 
Second, since in $\calA_{all}$, each user may output more than one report, the decomposition of uniform shuffling does not apply.
Third, as $P^G_i \neq P^G_j$ for $i,j \in [n]$ in general, one also needs to modify the uniform sampling assumptions made in \cite{soda-shuffling}.
We next study $P^G(t)$ which is vital to our analyses.

\subsection{Finite-Time Privacy Guarantees}
\label{sec:finite}
To characterize $P^G(t)$, let us take a closer look at the dynamics of message exchanges for $k$-regular (symmetric distribution) and ergodic graphs (stationary distribution).
Message exchanges on a $k$-regular graph can be tracked (in a probabilistic sense) at each time step due to its symmetric structure.
This allows us to calculate $P^G(t)$ exactly and provide a precise privacy analysis at any $t$.

On the other hand, generic ergodic graphs' $P^G(t)$ is dependent on the initial distribution especially when $t$ is finite.
We resort to giving an upper (worst-case in DP notion) bound in this case.
That is, we give an upper bound on $P^G(t)$ for a protocol that runs for t steps and then stops.
To derive the bound, it is convenient to consider it as a deviation from the stationary distribution, $\pi^G$, as detailed in the following. 

Recall from Section \ref{sec:rdmgrph} that any probability distribution at time step $t$ can be expanded with the eigenvectors of $AB^{-1}$: $P^G(t) = \sum_{i=1}^n c_i \alpha_i^te_i$ and $\pi^G = c_1 e_1$.
Then, using the fact that the orthogonal transformation preserves the inner product, the error $\Delta {P^G}^2(t)$ is
\begin{align}
\label{eq:pg_upper}
   \Delta {P^G}^2(t)&\coloneqq \sum\limits_{i\in[n]}{P^G_i}^2(t) - \sum\limits_{i\in[n]}{\pi^G_i}^2 \nonumber \\
    &= \sum_{i=1}^n c_i^2\alpha_i^{2t} - c_1^2 \nonumber \\
    &= \sum_{i=2}^n c_i^2\alpha_i^{2t} \leq (1-\alpha)^{2t},
\end{align}
following the arguments given in Equation \ref{eq:l2}.
Also from the arguments given in Section \ref{sec:rdmgrph}, we get $\Delta {P^G}^2(t) \lesssim \alpha/n^2$ when $t \simeq \alpha^{-1} \log n$.

The reason we compute $\sum\limits_{i\in[n]}{P^G_i}^2(t)$ will be made clear when we present the privacy theorems in the following, as we will see that they all depend on this quantity.
We begin by providing the privacy theorem for $\calA_{all}$ with stationary distribution next.
\subsection{ \texorpdfstring{$\calA_{all}$}{Aall} with Stationary Distribution}
\label{sec:all_protocol_proof}

\begin{thm}[``All'' protocol, Stationary distribution]
Let $\calA_{ldp}$ be a $\eps_0$-local randomizer. Let $\calA_{all}: \calD^n \to \calS^{(1)} \times \cdots \times \calS^{(n)}$ be the protocol as shown in Algorithm \ref{alg:a_all} sending all reports to the server.
Then, $\calA_{all}$ satisfies ($\eps,\delta+\delta_2$)-DP, with
\begin{align}
\label{eq:a_all_sta}
    \eps = \frac{(e^{\eps_0}-1)^2 e^{4 \eps_0}\eps_1^2} {2}
    +
    \eps_1\sqrt{2 (e^{\eps_0}-1)^2 e^{4 \eps_0} \log{\frac{1}{\delta}}},
\end{align}
$\eps_1 = \sqrt{(1-\frac{1}{n})\sum\limits_{i\in[n]} {P_i^G}^2} + \sqrt{\frac{\log(1/\delta_2)}{n}}$ and
$\sum\limits_{i\in[n]}{P^G_i}^2 \leq \sum\limits_{i\in[n]}{\pi^G_i}^2 + (1-\alpha)^{2t}$ when the protocol runs and stops at time step $t$. $\pi^G$ is the graph stationary distribution and $\alpha$ is the graph spectral gap.
Moreover, if $\calA_{ldp}$ is $(\eps_0,\delta_0)$-DP for $\delta_0 \leq \frac{(1 - e^{-\eps_0}) \delta_1}{4 e^{\eps_0} \left(2 + \frac{\ln(2/\delta_1)}{\ln(1/(1-e^{-5\eps_0}))}\right)}$, $\calA_{all}$ is $(\eps',\delta')$-DP with $\eps' = \frac{(e^{8\eps_0}-1)^2 e^{32 \eps_0}\eps_1^2} {2}
    +
    \eps_1\sqrt{2 (e^{8\eps_0}-1)^2 e^{32 \eps_0} \log{\frac{1}{\delta}}}$ and $\delta' = \delta + \delta_2+ n (e^{\eps'}+1) \delta_1$.
\label{thm:all_protocol}
\end{thm}

\noindent\textbf{Proof sketch.}
First, consider fixing (conditioning) the number of reports held by each user (e.g., user 1 holding 2 reports, user 2 holding zero report, and so on).
This conditioned distribution (realized by Algorithm~\ref{alg:a_fix}) may be seen as a distribution consisting of all permutations of data elements but with the number of reports held by each user fixed.
One may then reduce such a uniform permutation to swapping and uses a variant of the swapping technique \cite{soda-shuffling} to bound $\eps$.
Finally, a concentration bound on the distribution of report sizes, and a bound on $P^G$ using Equation \ref{eq:pg_upper} are placed to complete the proof.
The full proof is given in Section \ref{sec:proofs}.

\noindent\textbf{Interpretations.}
In Table \ref{tab:compare}, we compare the privacy amplification of network shuffling with other existing privacy amplification mechanisms.
We make comparisons assuming $\eps_0 > 1$ for convenience, and hide the dependence on polylogs of $\delta$ and $n$.
``No amplification'' means that only local randomizer is applied to each report.
Note that subsampling and uniform shuffling mechanisms require a centralized and trusted server to achieve the amplification.
Such an entity is powerful at breaking completely the link between the user and the reports.
Network shuffling works without this advantage, but yet is still capable of achieving amplification of $O(1/\sqrt{n})$ (albeit with a weaker exponential dependence on $\eps_0$).

We show in Figure \ref{fig:realdata_timestep} how the central DP guarantees change with respect to the number of communication round/time step per user.
Here, three real-world graphs (see Table \ref{tab:nws}) with around the same number of users ($n\simeq 2\text{-}3\times 10^4$) are used to make the comparisons.
$\alpha$ is calculated to be $\simeq 10^{-2}$, and we see that privacy guarantees converge at time step around $\alpha^{-1} \log n\simeq 10^{3}$. 
Note that the total communication overhead is $n$ multiplied by the number of time step, and there is no cover traffic as assumed in Section \ref{subsec:threat}.

\begin{table}
\begin{center}
\small
    \begin{tabular}{cccccc}
        \toprule
        & \multicolumn{3}{c}{social network} & comm & web \\
          & Facebook\cite{rozemberczki2019multiscale}& Twitch\cite{rozemberczki2019multiscale}& Deezer\cite{rozemberczki2019gemsec}& Enron\cite{DBLP:conf/ceas/KlimtY04}& Google\cite{DBLP:journals/im/LeskovecLDM09} \\
        \cmidrule(rl){1-1}\cmidrule(rl){2-4}\cmidrule(rl){5-5}\cmidrule(rl){6-6}
        $n$ & 22,470& 9,498& 28,281& 33,696& 855,802 \\
        $\Gamma^G$ & 5.0064& 7.5840& 3.5633& 36.866& 20.642 \\
        \bottomrule
    \end{tabular}
    \caption{Real-world network datasets. 
    The largest connected graphs are chosen when calculating the values of $n$ and $\Gamma^G$.}
    \label{tab:nws}
\end{center}
\end{table}

\begin{figure}
\begin{center}
    \includegraphics[width=0.45\textwidth]{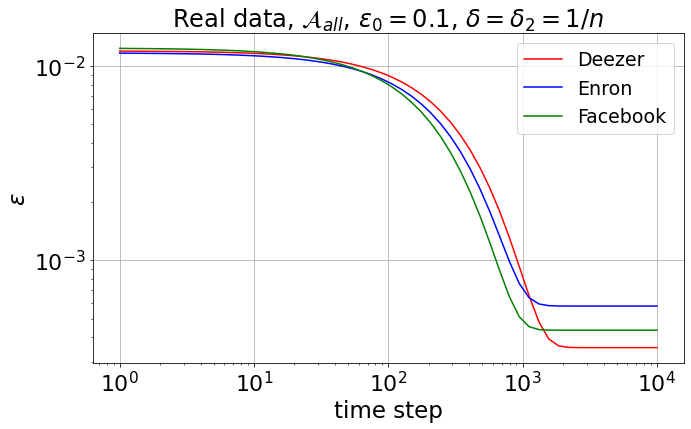}
\caption{Privacy-communication trade-offs, where it is shown that $\alpha^{-1}\log n$ time steps (around $10^3$ for selected datasets) are required for central $\eps$ to converge to the asymptotic value calculated under the stationary distribution assumption.}
\label{fig:realdata_timestep}
\end{center}
\end{figure}

\subsection{\texorpdfstring{$\calA_{all}$}{Aall} with Symmetric Distribution}
\begin{thm}[``All'' protocol, Symmetric distribution]
Let $\calA_{ldp}$ be a $\eps_0$-local randomizer. Let $\calA_{all}: \calD^n \to \calS^{(1)} \times \cdots \times \calS^{n}$ be the protocol as shown in Algorithm \ref{alg:a_all} sending all reports to the server. 
Then, $\calA_{all}$ satisfies ($\eps,\delta+\delta_2$)-DP, with
\begin{align*}
    \eps = \frac{(e^{\eps_0}-1)^2 e^{4 \eps_0}\eps_1^2} {2}
    +
    (e^{\eps_0}-1) e^{2 \eps_0}\eps_1\sqrt{2  \log{\frac{1}{\delta}}},
\end{align*}
$\eps_1 = \sqrt{(1-\frac{1}{n}){\rho^*}^2\sum\limits_{i\in[n]} {P_i^G}^2} + \sqrt{\frac{\log(1/\delta_2)}{n}}$ and $P^G$ the position probability distribution of any user when the protocol runs and stops at time step $t$.
$\rho^*$ is the ratio of the largest value of $P^G_i$ to the smallest non-zero $P^G_i$.
Moreover, if $\calA_{ldp}$ is $(\eps_0,\delta_0)$-DP for $\delta_0 \leq \frac{(1 - e^{-\eps_0}) \delta_1}{4 e^{\eps_0} \left(2 + \frac{\ln(2/\delta_1)}{\ln(1/(1-e^{-5\eps_0}))}\right)}$, $\calA_{all}$ is $(\eps',\delta')$-DP with $\eps' = \frac{(e^{8\eps_0}-1)^2 e^{32 \eps_0}\eps_1^2} {2}
    +
    \eps_1\sqrt{2 (e^{8\eps_0}-1)^2 e^{32 \eps_0} \log{\frac{1}{\delta}}}$ and $\delta' = \delta + \delta_2+ n (e^{\eps'}+1) \delta_1$.
\label{thm:all_protocol_sym}
\end{thm}
\noindent \textbf{Proof sketch.}
Our approach is similar to the proof of Theorem \ref{thm:all_protocol}.
The difference is in the conditioning of the output distribution with $L = \ell$ for some fixed $\ell \in [n]^n$ with $\sum_i \ell_i = n$.
For a pair of datasets $D$ and $D'$ differing in one element, we let the differing element be $x_1$ w.l.o.g..
Here, the conditioned output distribution has $x_1$ allocated at $i$ with probability $\frac{l_iP^G_i}{\sum\limits_{k\in[n]}l_kP^G_k}$, instead of uniform distribution like in Theorem \ref{thm:all_protocol}.
The rest of the calculation follows closely to those given in Section \ref{sec:proofs}.

\noindent 
\textbf{Trade-off between privacy and communication overheads.}
We show in Figure \ref{fig:regular_timestep} how the central DP guarantees change with respect to the number of communication rounds.
Perhaps not surprisingly, with larger $k$, the distribution ``mixes" faster as the walk has more choices of node to move to.
Subsequently, the privacy guarantees converge faster to the asymptotic value.
Note that we are tracing the random walk exactly in this case:
the walk exhibits non-monotonic behaviors in early times as it ``oscillates" between its neighbors without spreading out initially. 
This is in contrast to Figure \ref{fig:realdata_timestep} showing the upper bound of $\eps$ which is decreasing monotonously.

\begin{figure}
\begin{center}
    \includegraphics[width=0.45\textwidth]{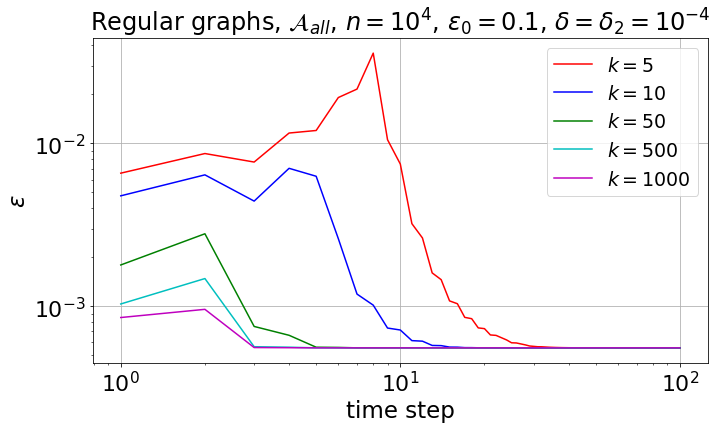}
\caption{The larger $k$ is, the faster $\eps$ converges to the asymptotic value for $k$-regular graphs.}
\label{fig:regular_timestep}
\end{center}
\end{figure}

\subsection{\texorpdfstring{$\calA_{single}$}{Asingle}}
The privacy theorems for $\calA_{single}$ with stationary and symmetric distributions are given below.
We begin with the stationary distribution:
\begin{thm}[``Single" protocol, Stationary distribution]
Let $\calA_{ldp}$ be a $\eps_0$-local randomizer. Let $\calA_{single}: \calD^n \to \calS^{(1)} \times \cdots \times \calS^{(n)}$ be the protocol as shown in Algorithm \ref{alg:a_single} sending single reports to the server.
Then, $\calA_{all}$ satisfies ($\eps,\delta$)-DP, with
\begin{align*}
    \eps = \frac{e^{2\eps_0}(e^{\eps_0}-1)^2}{2} \sum\limits_{i\in[n]} {P^G_i}^2 + e^{\eps_0}(e^{\eps_0}-1) \sqrt{2 \log{\frac{1}{\delta}} \sum\limits_{i\in[n]}{P^G_i}^2}
\end{align*}
and $\sum\limits_{i\in[n]}{P^G_i}^2 \leq \sum\limits_{i\in[n]}{\pi^G_i}^2 + (1-\alpha)^{2t}$ when the protocol runs and stops at time step $t$. $\pi^G$ is the graph stationary distribution and $\alpha$ is the graph spectral gap.
Moreover, if $\calA_{ldp}$ is $(\eps_0,\delta_0)$-DP for $\delta_0 \leq \frac{(1 - e^{-\eps_0}) \delta_1}{4 e^{\eps_0} \left(2 + \frac{\ln(2/\delta_1)}{\ln(1/(1-e^{-5\eps_0}))}\right)}$, $\calA_{all}$ is $(\eps',\delta')$-DP with $\eps' = \frac{e^{16\eps_0}(e^{8\eps_0}-1)^2}{2} \sum\limits_{i\in[n]} {P^G_i}^2 + e^{8\eps_0}(e^{8\eps_0}-1) \sqrt{2 \log{\frac{1}{\delta}} \sum\limits_{i\in[n]}{P^G_i}^2}$ and $\delta' = \delta + \delta_2+ n (e^{\eps'}+1) \delta_1$.
Particularly, when $\eps_0 \leq 1$, we obtain
$\eps' = 800 \eps_0^2 \sum\limits_{i\in[n]} {P^G_i}^2 + 40\eps_0 \sqrt{2 \log{\frac{1}{\delta}} \sum\limits_{i\in[n]}{P^G_i}^2}$
 and $\delta' = \delta + \delta_2+ n (e^{\eps'}+1) \delta_1$.
\label{thm:single_protocol}
\end{thm}

For the symmetric distribution, the privacy guarantees for the ``single" protocol turn out to be almost the same as Theorem \ref{thm:single_protocol}:
\begin{thm}[``Single" protocol, Symmetric distribution]
Let $\calA_{ldp}$ be a $\eps_0$-local randomizer.
$\calA_{single}$ satisfies ($\eps,\delta$)-DP, with $\eps$ equal to the one given in Theorem \ref{thm:single_protocol}, except that $P^G$ is the position probability distribution of any user when the protocol runs and stops at time step $t$.
Moreover, if $\calA_{ldp}$ is $(\eps_0,\delta_0)$-DP for $\delta_0 \leq \frac{(1 - e^{-\eps_0}) \delta_1}{4 e^{\eps_0} \left(2 + \frac{\ln(2/\delta_1)}{\ln(1/(1-e^{-5\eps_0}))}\right)}$, $\calA_{single}$ is $(\eps',\delta')$-DP with $\eps',\delta'$ also equal to the ones given in Theorem \ref{thm:single_protocol} except that $P^G$ is the position probability distribution of any user at time step $t$.
\label{thm:single_protocol_sym}
\end{thm}
\label{sec:proof_ampl_dist}
\noindent \textbf{Proof sketch.}
In order to prove the privacy guarantees of $\calA_{single}$ (Algorithm~\ref{alg:a_single}), we utilize random replacement \cite{DBLP:conf/nips/BalleKMTT20}: we reduce $\calA_{single}$ to an algorithm which works as follows: Substitute the first element in the dataset with a pre-determined element, and choose an element in the dataset (according to a certain probability distribution) to substitute it with the original first element.
Then, all elements are randomized with the local randomizer. 
The rest of the calculation follows similarly to that given in Section \ref{sec:proofs}.
 \footnote{The communication overheads of $\mathcal{A}_{single}$ have the same trend as those of $\mathcal{A}_{all}$ and are therefore not shown for brevity.}

\subsection{Numerical Analyses}
Here, we provide several numerical analyses based on our privacy theorems.
In the following, unless stated otherwise, we always assume that the network shuffling protocol runs and stop at time $t= \lfloor\alpha^{-1}\log n\rceil$ (mixing time), and present the results correct at this $t$.
Table \ref{tab:nws} shows values of $n$ and $\Gamma^G$'s for five real-world network datasets.
We see that social networks, which are our main use case, have reasonably regular structure ($\Gamma^G \lesssim 10$) compared to other networks, e.g., Google web graph.
The dependence of the amplified $\eps$ on $\eps_0$ for the real-world datasets is shown in Figure \ref{fig:th44}.
It can be seen that, nevertheless, population size matters the most as Google with $n\simeq 10^6$ yields the most significant privacy amplification.

Figure \ref{fig:real_all_single} compares the level of privacy amplification of $\mathcal{A}_{all}$ and $\mathcal{A}_{single}$.
Plots are made using two datasets (Twitch and Google) which differ significantly in $n$ (9,498 and 855,802 respectively).
Furthermore, in Figure \ref{fig:th44th}, we show the dependence of the amplified $\eps$ on $\eps_0$ varying the relevant parameters ($n$, $\Gamma^G$ and protocol) without assumptions on the underlying dataset, but at the limit of stationary distribution, in constrast to previous analyses.

\begin{figure}
\begin{center}
    \includegraphics[width=0.45\textwidth]{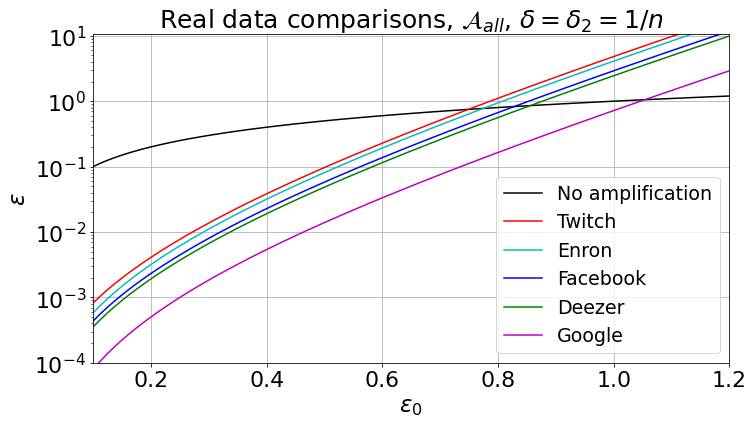}
\caption{
The Google dataset with the largest population achieves the most significant privacy amplification for $\eps_0$ ranging from 0.1 to 1.2 calculated assuming the $\mathcal{A}_{all}$ protocol.}
\label{fig:th44}
\end{center}
\end{figure}

\begin{figure}
\begin{center}
    \includegraphics[width=0.45\textwidth]{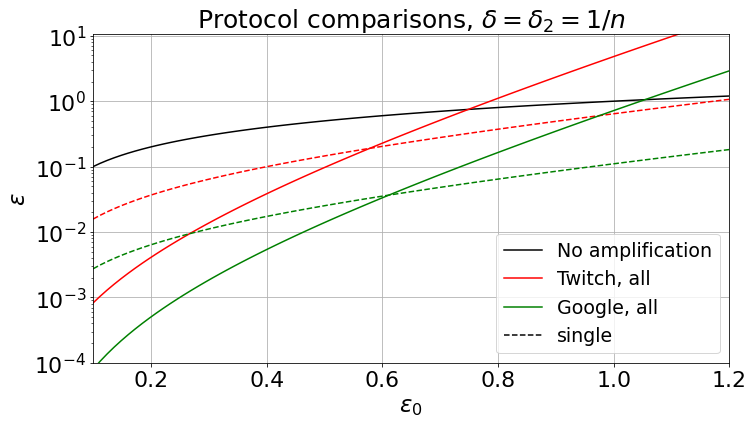}
\caption{Comparison of central $\eps$ under protocol $\mathcal{A}_{all}$ (continuous line) and $\mathcal{A}_{single}$ (dashed line), showing that $\mathcal{A}_{single}$ achieves larger privacy amplification at large $\eps_0$.}
\label{fig:real_all_single}
\end{center}
\end{figure}

\begin{figure}
\begin{center}
    \includegraphics[width=0.45\textwidth]{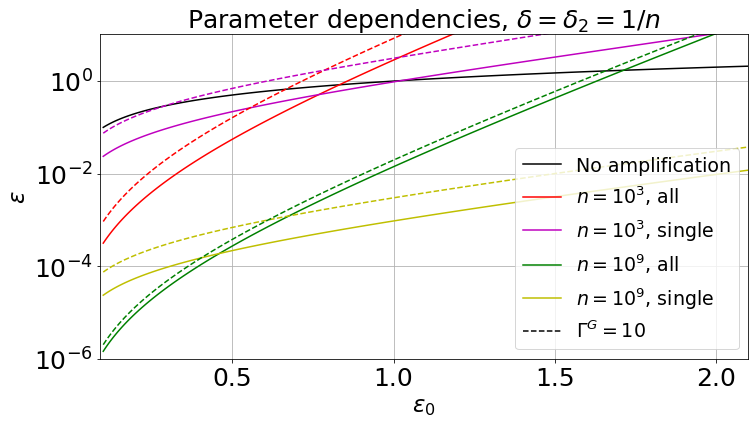}
\caption{
Stationary distribution's parameter dependencies, comparing central $\eps$ for various $\Gamma^G$, $n$ and protocols for $\eps_0$ ranging from 0.2 to 2.0.
The black line corresponds to $\eps=\eps_0$, i.e., no amplification. Other colored continuous and dashed lines correspond to scenarios where the underlying graphs are such that $\Gamma^G=1$ (regular) and $\Gamma^G=10$, respectively.}
\label{fig:th44th}
\end{center}
\end{figure}

Notice (e.g., in Figure \ref{fig:real_all_single}) that at large $\eps_0$, using $\mathcal{A}_{single}$ gives better privacy amplification.
However, we should not conclude that utility-wise $\mathcal{A}_{single}$ is always better at large $\eps_0$.
This is because, as mentioned in Section \ref{subsec:protocol}, $\mathcal{A}_{single}$ does not send all reports truthfully (user not holding any report sends a dummy report; user holding multiple reports sends only a report), leading to potential utility losses.
To demonstrate this, we study a privacy-utility trade-off scenario by performing mean estimation on the Twitch dataset following a setup similar to \cite{DBLP:conf/nips/ChenKO20}: we generate $d$-dimensional synthetic samples independently but non-identically; we set $z_1, \dots, z_{n/2}\overset{i.i.d.}{\sim} N(1,1)^{\otimes d}$ and $z_{n/2}, \dots, z_{n}\overset{i.i.d.}{\sim} N(10,1)^{\otimes d}$.
Each sample is normalized by setting $x_i = z_i/\|z_i\|_2$.
Additionally, we generate dummy sample (as required by $\mathcal{A}_{single}$) by setting $z\overset{i.i.d.}{\sim} N(5,1)^{\otimes d}$.
$d$ is set to be 200.

The \texttt{PrivUnit} \cite{DBLP:journals/corr/abs-1812-00984} algorithm is applied to perturb each report for obtaining $\eps_0$-LDP guarantees.
The utility of mean estimation is quantified by the mean squared loss or the $l_2$ error of $x$.
The number of dummy data is determined as the expected number of user holding less than a sample (7,080 for Twitch).

In Figure \ref{fig:util}, we plot the relationship between the central $\eps$ and the expected squared error.
This is done by sampling a few points of $\eps_0$, applying \texttt{PrivUnit} to the data, calculating the central $\eps$ and the expected squared error according to the protocols.
It can be observed that at least in the studied parameter region, for any fixed value of $\eps$, the expected squared error of $\mathcal{A}_{all}$ is consistently smaller than that of $\mathcal{A}_{single}$, serving as a counter-example to the argument that $\mathcal{A}_{single}$ is better at large privacy region. 
In general, the utility-privacy trade-offs of $\mathcal{A}_{all}$ and $\mathcal{A}_{single}$ are scenario and dataset-dependent.

\begin{figure}
\begin{center}
    \includegraphics[width=0.45\textwidth]{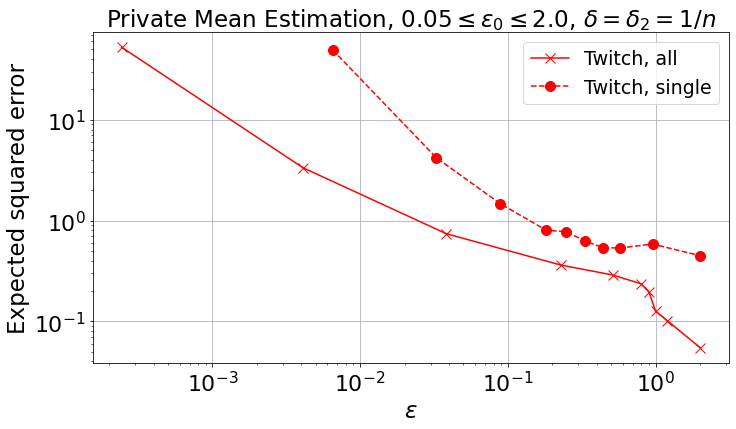}
\caption{The $\mathcal{A}_{all}$ protocol (continuous line) consistently performs better utility-wise than the $\mathcal{A}_{single}$ protocol (dashed line) for various values of $\eps$ when performing private mean estimation on the Twitch dataset.}
\label{fig:util}
\end{center}
\end{figure}

\section{Detailed Privacy Analysis of \texorpdfstring{$\calA_{all}$}{Aall} with Stationary Distribution}
\label{sec:proofs}
In this section, we give a detailed proof of the privacy theorem of $\calA_{all}$ with stationary distribution, Theorem \ref{thm:all_protocol}.
The proofs of other scenarios are omitted due to space limit and will be given elsewhere.

Let $L = (L_1, \ldots, L_n)$ be the report sizes, i.e., the number of reports received by each user, $L_i = |S_i|$ for $i \in [n]$ as in Algorithm~\ref{alg:a_fix}, $\calA_{fix}(D)$.
The output distribution of $\calA_{fix}(D)$ may be considered to be a distribution conditioned on $L = \ell$ for some fixed $\ell \in [n]^n$ with $\sum_i \ell_i = n$.
Fixing the position probability distribution, the output distribution is then the same as the one produced by Algorithm~\ref{alg:a_fix} with report sizes $\ell$, but consisting of all permutations $\pi(D)$ of the original dataset $D$.
This can be viewed as a variant of shuffle model with fixed report sizes, and can be analyzed by reducing random permutation/uniform shuffling to swapping \cite{soda-shuffling}.
For a pair of datasets $D$ and $D'$ differing in the first record, according to this reduction, it is sufficient to analyze $\calA_{fix}(\sigma(D), \ell)$, where $\sigma(D)$ is a procedure swapping $x_1$ with $x_i$ for $i$ uniformly sampled from $[n]$.
We henceforth prove the following theorem.

\begin{algorithm}[ht]
\caption{$\calA_{fix}$: Local responses with fixed report sizes}
\begin{algorithmic}[1]
    \STATE \textbf{Input:} Dataset $D=x_{1:n}$, report sizes $\ell \in [n]^n$ with $\sum_i \ell_i = n$, local randomizer $\calA_{ldp}: \calS^{(1)} \times \cdots \times \calS^{(i-1)} \times \calD \rightarrow \calS^{(i)}$ for $i\in [n]$. %
    \STATE $j \gets 1$
    \FOR{$i \in [n]$}
    \STATE $S_i \gets \{\}$
    \ENDFOR
    \FOR{$i \in [n]$}
        \IF{$\ell_i = 0$}
            \STATE do nothing
        \ELSE
            \FOR{$k \in \{j, \ldots, j + \ell_i - 1\}$}
                \STATE $s_k \gets \calA_{ldp}(x_k)$
                \STATE add  $s_k$ to $S_i$
            \ENDFOR
            \STATE $j \gets j + \ell_i$
        \ENDIF
    \ENDFOR
    \STATE \textbf{return} sequence $S_{1:n}$
\end{algorithmic}
\label{alg:a_fix}
\end{algorithm}

\begin{thm}\label{thm:fix}
Let $\calA^{(i)}_{ldp}: \mathcal{S}^{(1)} \times \cdots \times \mathcal{S}^{(i - 1)} \times \calD \rightarrow \mathcal{S}^{(i )}$ for $i \in [n]$ be an $\eps_0$-DP local randomizer.
Let $\ell \in [n]^n$ s.t. $\sum_i \ell_i = n$.
We also define $\sigma(D)$ to be the swapping operation on dataset $D = \{x_1, \ldots, x_n\}$, where $x_1$ is swapped with $x_i$ for $I$ uniformly sampled from $[n]$.
Then, $M(D) = \calA_{fix}(\sigma(D), \ell)$ is $(\eps,\delta+\delta_2)$-DP at index $1$ with $\eps = \frac{(e^{\eps_0}-1)^2 e^{4 \eps_0}\eps_1^2} {2}
    +
    \eps_1\sqrt{2 (e^{\eps_0}-1)^2 e^{4 \eps_0} \log{\frac{1}{\delta}}}$, where $\eps_1 = \sqrt{(n-1)\sum\limits_{i\in[n]} {P_i^G}^2} + \sqrt{\frac{\log(1/\delta_2)}{n}}$ and
$\sum\limits_{i\in[n]}{P^G_i}^2 \leq \sum\limits_{i\in[n]}{\pi^G_i}^2 + (1-\alpha)^{2t}$ when the protocol runs and stops at time step $t$. $\pi^G$ is the graph stationary distribution and $\alpha$ is the graph spectral gap.
Moreover, if $\calA_{ldp}$ is $(\eps_0,\delta_0)$-DP for $\delta_0 \leq \frac{(1 - e^{-\eps_0}) \delta_1}{4 e^{\eps_0} \left(2 + \frac{\ln(2/\delta_1)}{\ln(1/(1-e^{-5\eps_0}))}\right)}$, $M(D) = \calA_{fix}(\sigma(D), \ell)$ is $(\eps',\delta')$-DP at index $1$ with $\eps' = \frac{(e^{8\eps_0}-1)^2 e^{32 \eps_0}\eps_1^2} {2}
    +
    \eps_1\sqrt{2 (e^{8\eps_0}-1)^2 e^{32 \eps_0} \log{\frac{1}{\delta}}}$ and $\delta' = \delta + \delta_2+ n (e^{\eps'}+1) \delta_1$.

\end{thm}
\begin{proof}
We begin with the case where the local randomizer $\calA_{ldp}$ is $\eps_0$-DP.
Since the output of the algorithm is a sequence $S$ of length $n$, we analyze the privacy of each of the elements of $S$ and apply adaptive composition on them.
We first let $y_1, \ldots, D_{n}$ denote $n$ disjoint subsets of $D$ where the size is $|y_i| = \ell_i$ for $i\in [n]$.
Also let $\bar{D}_1, \ldots, \bar{D}_{n}$ denote $n$ disjoint subsets of $D$ after the swapping operation ($\sigma(D)$), where $\bar{D}_1=(\bar{x_1}, \ldots, \bar{x}_{\ell_1})$ and so on.
Finally, we let $D'$ be the dataset differing in the first element compared to $D$.

We consider a mechanism, $\calB^{(i)}$ for $i \in [n]$, that takes the whole (swapped) dataset and previous outputs $s_{1:i-1}$ as input, and outputs $s_i$ with internal randomness independent of $s_{1:i-1}$.
For convenience, we write $\calA^{(i)}( s_{1:i-1};x_i)= (\calA_{ldp}(\bar{x}_{k}; s_{1:i-1}),\ldots)$ where $(\ldots)$ includes all $\bar{x}_k \in \bar{D}_i$.

$\calB^{(i)}$ can be seen as a mixture of two distributions conditioned on whether $x_1 \in \bar{D}_i$.
For notational convenience, let $\mu$ be the probability distribution of $\calB^{(i)}$. Rewriting $\mu$ as $\mu = (1-p_i)\mu_0 + p_i \mu_1$, where $p_i = \Pr[x_1 \in \bar{D}_i | s_{1:i-1}]$, we have
\begin{align*}
    &\Pr[\calB^{(i)}(s_{1:i-1}; D)  = s_i]
    \\
    & =
    (1-p_i) \Pr[\calA^{(i)}(s_{1:i-1}; x_i) = s_i | s_{1:i-1}, x_1 \notin \bar{D}_i]\\
    &\qquad + \frac{p_i}{\ell_i} \sum_{d \in x_i} \Pr[\calA^{(i)}(s_{1:i-1}; x_i \cup \{x_1\} \setminus \{d\}) = s_i | s_{1:i-1}, x_1 \in \bar{D}_i].
    \enspace
\end{align*}
For dataset $D'$, the corresponding probability distribution is $\mu' = (1-p_i' )\mu_0' + p_i' \mu_1'$, with $p_i', \mu_0', \mu_1'$ the corresponding quantities.

Next, we wish to show that $\mu$ and $\mu'$ are overlapping mixtures \cite{BBG18}.
For $i \in [n]$, $\mu_0= \mu'_0$ because $x_1$ is not involved in $\calB^{(i)}$.
Then, we move the $p_i'$ part of the first component of $\mu'$ to the second component:
$$\mu' = (1-p_i ) \mu_0 + p_i  \left(\frac{p_i'}{p_i} \mu_1' + \left(1 - \frac{p_i'}{p_i}\right)\mu_0\right)
= (1-p_i  ) \mu_0 + p_i  \mu_1'',$$
where, w.l.o.g, $p_i \geq p_i'$ has been assumed.
This indicates that $\mu$ and $\mu'$ are overlapping mixtures.

From the fact that $\calA_{ldp}$ is $\eps_0$-LDP, we have $\mu_0 \approxeq_{(\eps_0,0)} \mu_1$ and $\mu_0=\mu_0' \approxeq_{(\eps_0,0)} \mu_1'$, as well as $\mu_1 \approxeq_{(\eps_0,0)} \mu_1'$ (because $\mu \approxeq_{(\eps_0,0)} \mu'$ and $\mu_0=\mu_0'$). Therefore, by joint convexity, $\mu_1 \approxeq_{(\eps_0,0)} \mu_1''$.
With these, we get, from the advanced joint convexity of overlapping mixtures (Theorem 2 in \cite{BBG18}), 
$\mu \approxeq_{(\log\left(p_i (e^{\eps_0} - 1)+1\right),0)} \mu'.$

We proceed to bound $p_i$.
Since $\Pr[x_1 \in \bar{D}_i]=\frac{l_i}{n}$ for $i\in [n]$,
\begin{align*}
    p_i
    &=
    \Pr[x_1 \in \bar{D}_i | s_{1:i-1}]
    \\
    &=
    \frac{\Pr[s_{1:i-1} | x_1 \in \bar{D}_i] \Pr[x_1 \in \bar{x_i}]}{\Pr[s_{1:i-1}]}
    \\
    &=
    \frac{\ell_i}{n} \frac{\Pr[s_{1:i-1} | x_1 \in \bar{D}_i]}{\sum_{k \in [n]} \Pr[s_{1:i-1} | x_1 \in \bar{D}_k] \Pr[x_1 \in \bar{D}_k]}
    \\
    &=
    \frac{\ell_i}{\sum_{k \in [n]} \ell_k \frac{\Pr[s_{1:i-1} | x_1 \in \bar{D}_k]}{\Pr[s_{1:i-1} | x_1 \in \bar{D}_i]}}.
    \enspace
\end{align*}
Now, we note that for $k < i$, the subset $\bar{D}_{1:i-1}$ conditioned on $x_1 \in \bar{D}_k$ and the subset $\bar{D}_{1:i-1}$ conditioned on $x_1 \in \bar{D}_i$ differs in two positions at most.
This leads to $\frac{\Pr[s_{1:i-1} | x_1 \in \bar{D}_k]}{\Pr[s_{1:i-1} | x_1 \in \bar{D}_i]} \leq e^{-2\epsilon_0}$ for local randomizer satisfying $\eps_0$-LDP.
Similarly, for $k>i$, we have the two aforementioned subsets differing in one position at most, leading to $\frac{\Pr[s_{1:i-1} | x_1 \in \bar{D}_k]}{\Pr[s_{1:i-1} | x_1 \in \bar{D}_i]} \leq e^{-\eps_0}$.
Hence, we can write 
\begin{align*}
    \sum_{k \in [n]} \ell_k \frac{\Pr[s_{1:i-1} | x_1 \in \bar{D}_k]}{\Pr[s_{1:i-1} | x_1 \in \bar{D}_i]}
    &\geq
    \ell_i + e^{-2 \eps_0} \sum_{k < i} \ell_k + e^{-\eps_0} \sum_{k > i} \ell_k,
    \enspace
\end{align*}
which is larger than or equal to $e^{-2 \eps_0} n$.
The bound on $p_i$ is therefore $p_i \leq e^{2 \eps_0}l_i/ n$.
Furthermore, from the advanced joint convexity of overlapping mixtures, we get that $\calB^{(i)}$ is $\eps_i$-DP with $\eps_i \leq \log(1 + e^{2 \eps_0} (e^{\eps_0} - 1) \ell_i / n)$.

Now, we use Equation \ref{eqn:heta} to calculate the overall $(\eps,\delta)$ of $\calB^{(1)}, \ldots, \calB^{(n)}$:
\begin{align*}
    \eps
    &=
    \sum\limits_{i\in[n]} \frac{(e^{\eps_i} - 1) \eps_i}{e^{\eps_i} + 1} + \sqrt{2 \log{\frac{1}{\delta}} \sum\limits_{i\in[k]}\eps_i^2}
    \\
    &\leq
    \sum\limits_{i\in[n]}
    \frac{e^{2 \eps_0} (e^{\eps_0} - 1) \ell_i / n \cdot \log(1 + e^{2 \eps_0} (e^{\eps_0} - 1) \ell_i / n)}{2 + e^{2 \eps_0} (e^{\eps_0} - 1) \ell_i / n}
    \\
    &+
    \sqrt{2 \log\frac{1}{\delta} \sum\limits_{i\in[n]} \left(\log(1 + e^{2 \eps_0} (e^{\eps_0} - 1) \ell_i / n) \right)^2}
    \\
    &\leq
    \frac{(e^{\eps_0}-1)^2 e^{4 \eps_0} \ltwo{\ell}^2}{2 n^2}
    +
    \sqrt{\frac{2 (e^{\eps_0}-1)^2 e^{4 \eps_0} \ltwo{\ell}^2}{n^2} \log{\frac{1}{\delta}}},
\end{align*}
where we have used $\log(1+x) \leq x$ and $\frac{1}{2+x} \leq 1/2$ for $x\geq 0$ in the second line.

We next give a bound on $\ltwo{l}$.
Using Lemma \ref{lem:allocate}, we obtain with probability at least $1-\delta_2$,

\begin{align*}
    \frac{\ltwo{l}}{n} \leq \sqrt{(1-\frac{1}{n})\sum\limits_{i\in[n]} {P_i^G}^2} + \sqrt{\frac{\log(1/\delta_2)}{n}} \enspace.
\end{align*}
Note that $P^G$ can be upper bounded using Equation \ref{eq:pg_upper}.
Combining everything, we have 
\begin{align}
    \eps \leq \frac{(e^{\eps_0}-1)^2 e^{4 \eps_0}\eps_1^2} {2}
    +
    (e^{\eps_0}-1) e^{2 \eps_0}\eps_1\sqrt{2  \log{\frac{1}{\delta}}},
    \label{eq:eps_fix_puredp}
\end{align}
where $\eps_1 = \sqrt{(1-\frac{1}{n})\sum\limits_{i\in[n]} {P_i^G}^2} + \sqrt{\frac{\log(1/\delta_2)}{n}}$, as desired.

We next prove the case where the local randomizer is $(\eps_0,\delta_0)$-DP.
Using Lemma \ref{lem:cheu}, we know that for any $\calA$ satisfying $(\eps_0,\delta_0)$ with $\delta_0 \leq \frac{(1 - e^{-\eps_0}) \delta_1}{4 e^{\eps_0} \left(2 + \frac{\ln(2/\delta_1)}{\ln(1/(1-e^{-5\eps_0}))}\right)}$, there exists $\tilde{\calA}_{ldp}$ that is $8\eps_0$-DP for each element $x$ and satisfies $TV(\calA_{ldp}(x),\tilde{\calA}_{ldp}(x))\leq \delta_1$. 
A union bound over the whole dataset $D$ of size $n$ then gives a bound between $\calA_{fix}$ and $\tilde{\calA_{fix}}$:
$$TV(\calA_{fix}(D),\tilde{\calA}_{fix}(D))\leq n\delta_1.$$
Using the argument for the pure DP case, we see that an $8\eps_0$-DP local randomizer $\calA_{ldp}$ satisfies $(\eps',\delta)$-DP at the first index where $\eps'$ is given as in Equation \ref{eq:eps_fix_puredp} but with $\eps_0$ replaced with $8\eps_0$.
Finally, using Proposition 3 from \cite{DBLP:conf/icml/WangFS15}, we obtain that $\calA_{fix}$ satisfies $(\eps',\delta')$-DP with $\delta' = \delta + \delta_2 + n(e^{\eps'}+1)\delta_1$.
\end{proof}
One then applies joint convexity to the output distribution conditioned by $L$ to complete the proof of Theorem \ref{thm:all_protocol}.

\section{Conclusion}
In this work, we have rethought the shuffling mechanism and proposed network shuffling to show that privacy amplification without a trusted and centralized entity is possible.
While this work is largely theoretical, the research question is inspired by a practical observation: Can we achieve privacy amplification without a centralized, trusted entity due to known security issues and constraints? 

Let us point out several future directions extending our work.
Our privacy accounting may be further tightened with more advanced techniques. 
Besides, integrating other variants of walk like lazy random walk are of practical importance.
Note also that our scenario still requires a central authority to initiate the task of collecting data, determine the number of communication rounds and network topology, etc.
It may be possible to achieve full decentralization where the clients themselves set up the task collaboratively through the use of, e.g., a distributed ledger \cite{DBLP:journals/ftml/KairouzMABBBBCC21}.

As a final remark, let us emphasize that, although we are contributing to realizing privacy amplification with practicality in mind, there is still a lot to do for deploying real-world distributed computation systems, as we have revealed throughout this work.
We hope that our work can serve as an inspiration for researchers and practitioners to design new privacy mechanisms tackling other practical problems and invent new architectures as well as other novel methods to realize the desired privacy properties in the end-to-end systems.

\bibliographystyle{ACM-Reference-Format}
\balance
\bibliography{draft}

\end{document}